\documentclass[journal]{IEEEtran}
\usepackage{algorithm, algpseudocode}
\usepackage{authblk, cite, enumitem, graphicx, hyperref, pgfplots, stmaryrd, subfig}
\usepackage{amsfonts, amsmath, amssymb, amsthm}
\usepackage{CJKutf8, dsfont}
\usepackage[normalem]{ulem}
\newlist{todolist}{itemize}{2}
\setlist[todolist]{label=$\square$}
\hypersetup{colorlinks=true, citecolor=blue(ryb), linkcolor=blue(ryb), urlcolor=cyan}
\PassOptionsToPackage{hyphens}{url}
\definecolor{azure}{rgb}{0.0, 0.5, 1.0}
\definecolor{frenchblue}{rgb}{0.0, 0.45, 0.73}
\definecolor{forestgreen(traditional)}{rgb}{0.0, 0.27, 0.13}
\definecolor{mygreen}{rgb}{0.09, 0.45, 0.27}
\definecolor{myblue}{rgb}{0.2383,0.5195,0.7734}
\definecolor{mygreen}{rgb}{0.6445,0.9297,0.0039}
\definecolor{darklavender}{rgb}{0.45, 0.31, 0.59}
\definecolor{americanrose}{rgb}{1.0, 0.01, 0.24}
\definecolor{pigblue}{rgb}{0.2, 0.2, 0.6}
\definecolor{blue(ryb)}{rgb}{0.01, 0.28, 1.0}
\definecolor{amethyst}{rgb}{0.6, 0.4, 0.8}
\definecolor{deepmagenta}{rgb}{0.8, 0.0, 0.8}
\definecolor{carminered}{rgb}{1.0, 0.0, 0.22}
\definecolor{iris}{rgb}{0.35, 0.31, 0.81}

\newtheorem{theorem}{\hspace{0pt}\bf Theorem}

\newtheorem{definition}{\hspace{0pt}\bf Definition}

\newcommand*{\herm}{{\mkern-1.5mu\mathsf{H}}}

\newcommand{\reals}{{\mbox{\bf R}}}

\newcommand{\diag}{\mathop{\bf diag}}

\newcommand{\argmin}{\mathop{\rm argmin}}

\def \SNR     {\text{\normalfont SNR}   }

\def \cov     {\text{\normalfont cov}  }

\def \diag    {\text{\normalfont diag} }

\def \reals    {{\mathbb R}}

\def\mbC{{\ensuremath{\mathbb C}}}

\def\ccalH{{\ensuremath{\mathcal H}}}

\def\ccalL{{\ensuremath{\mathcal L}}}
\def\ccalM{{\ensuremath{\mathcal M}}}
\def\ccalN{{\ensuremath{\mathcal N}}}
\def\ccalO{{\ensuremath{\mathcal O}}}

\def\ccal0{{\ensuremath{\mathcal 0}}}
\def\bbzero{{\ensuremath{\boldsymbol 0}}}
\def\bbA{{\ensuremath{\boldsymbol A}}}
\def\bbB{{\ensuremath{\boldsymbol B}}}

\def\bbE{{\ensuremath{\boldsymbol E}}}

\def\bbI{{\ensuremath{\boldsymbol I}}}

\def\bbX{{\ensuremath{\boldsymbol X}}}
\def\bbY{{\ensuremath{\boldsymbol Y}}}

\def\bba{{\ensuremath{\boldsymbol a}}}
\def\bbb{{\ensuremath{\boldsymbol b}}}
\def\bbc{{\ensuremath{\boldsymbol c}}}

\def\bbe{{\ensuremath{\boldsymbol e}}}

\def\bbh{{\ensuremath{\boldsymbol h}}}

\def\bbu{{\ensuremath{\boldsymbol u}}}
\def\bbv{{\ensuremath{\boldsymbol v}}}
\def\bbs{{\ensuremath{\boldsymbol s}}}

\def\bbx{{\ensuremath{\boldsymbol x}}}
\def\bby{{\ensuremath{\boldsymbol y}}}

\def\bb0{{\ensuremath{\boldsymbol 0}}}
\def\bbgamma{\boldsymbol{\gamma}}

\def\bbGamma{\boldsymbol{\Gamma}}

\def\bbTheta{\boldsymbol{\Theta}}

\def\bbSigma{\boldsymbol{\Sigma}}

\def\hbgamma{\hat\bbgamma}

\def\hhatgamma{\hat\gamma}

\definecolor{crimson}{rgb}{0.86, 0.08, 0.24}
\definecolor{scarlet}{rgb}{1.0, 0.13, 0.0}
\definecolor{hookersgreen}{rgb}{0.0, 0.44, 0.0}
\definecolor{cultramarine}{rgb}{0.25, 0.0, 1.0}

\newcommand\revised[1]{#1}

\def \inv {{-1}}

\def \bS {\bbSigma}
\def \bSb {\bbSigma_{\backslash i}}
\def \bSbi {\bSb^{-1}}
\def \per {\ensuremath{P_\textrm{ER}}}
\def \pmd {\ensuremath{P_\textrm{MD}}}
\def \pfa {\ensuremath{P_\textrm{FA}}}

\title{Robust Activity Detection for Massive Random Access}
\author{Xinjue~Wang,~\IEEEmembership{Graduate Student Member,~IEEE,}
        Esa~Ollila,~\IEEEmembership{Senior Member,~IEEE,}
        and~Sergiy~A.~Vorobyov,~\IEEEmembership{Fellow,~IEEE}
\thanks{This work was supported by the Research Council of Finland under Grant 359848, and Grant 357715.
An earlier version of this paper was presented at the IEEE International Conference on Acoustics, Speech and Signal Processing (ICASSP), Hyderabad, India, 2025 [DOI: 10.1109/ICASSP49660.2025.10889433].
\textit{(Corresponding author: Esa Ollila.)}

The authors are with the Department of Information and Communications Engineering, Aalto University, 02150 Espoo, Finland.}}

\begin{document}
\maketitle
\begin{abstract}
Massive machine-type communications (mMTC) are fundamental to the Internet of Things (IoT) framework in future wireless networks, involving the connection of a vast number of devices with sporadic transmission patterns. Traditional device activity detection (AD) methods are typically developed for Gaussian noise, but their performance may deteriorate when these conditions are not met, particularly in the presence of heavy-tailed impulsive noise.   In this paper, we propose robust statistical techniques for AD that do not rely on the Gaussian assumption and replace the Gaussian loss function with robust loss functions that can effectively mitigate the impact of heavy-tailed noise and outliers. First, we prove that the coordinate-wise (conditional) objective function is geodesically convex and derive a fixed-point (FP) algorithm for minimizing it, along with convergence guarantees. Building on the FP algorithm, we propose two robust algorithms for solving the full (unconditional) objective function: a coordinate-wise optimization algorithm (RCWO) and a greedy covariance learning-based matching pursuit algorithm (RCL-MP). Numerical experiments demonstrate that the proposed methods significantly outperform existing algorithms in scenarios with non-Gaussian noise, achieving higher detection accuracy and robustness. 
\end{abstract}

\section{Introduction}
\label{sec:intro}
\IEEEPARstart{M}{assive} machine-type communications (mMTC) represent a cornerstone of future wireless communication networks, particularly within the Internet of Things~(IoT) framework~\cite{liu2018sparse}. 
This paradigm involves connecting a vast number of devices such as sensors, machines, and robots, leading to immense data traffic. 
Unlike traditional communication systems, mMTC is characterized by high device density and sporadic transmission patterns, where only a small subset of devices is active at any given time, typically triggered by external events~\cite{bockelmann2018towards, chen2020massive, dawy2016toward}.
Consequently, dynamically and accurately detecting the set of active devices becomes a challenging and critical task.

\revised{
Conventional wireless systems rely on grant-based random access protocols where devices request transmission opportunities using a limited pool of orthogonal pilot sequences. 
When transmitting data, an active device randomly selects one of these sequences and sends it as a resource allocation request.
However, as the number of simultaneously active devices exceeds the number of available orthogonal sequences, inevitable collisions occur.
Resolving these collisions requires multiple signaling rounds between devices and the base station (BS), introducing delays that contradict the low-latency requirements of mMTC applications~\cite{liu2018massiveAMP, fengler2021nonTIT, ganesan2021clustering, chen2021phase}.

Grant-free random access addresses collision issues in mMTC by preassigning unique non-orthogonal pilot sequences to each device~\cite{shahab2020grant}.
Active devices transmit their pilots simultaneously without awaiting permission, eliminating signaling overhead and reducing latency for sporadic traffic scenarios~\cite{fengler2021sparse_thesis, fengler2021nonTIT}.
In both grant-based and grant-free random access schemes, accurate device \textit{activity detection} (AD) is important for establishing successful communications between the devices and the BS~\cite{kang2021minimum, kang2022scheduling, liu2018massiveAMP}.
In grant-free framework, AD is more challenging due to non-orthogonal pilot sequences and the lack of coordination among devices \cite{liu2018massiveAMP}.
In this paper, we focus on maintaining detection reliability under non-Gaussian noise conditions.
}

Device AD can be formulated as a compressive sensing problem due to the inherent sparsity arising from the small number of active devices relative to the total number of devices~\cite{liu2018sparse, chen2018sparse}. 
Many algorithms for the sparse signal recovery problem have been developed under the assumption that the data follow Gaussian distribution, justified by the central limit theorem and mathematical tractability. 
This Gaussianity assumption is widely made in sparse recovery methods such as approximate message passing (AMP)~\cite{donoho2009AMP, liu2018sparse, liu2018massiveAMP}, sparse Bayesian learning (SBL)~\cite{tipping2001sparsetruSBL}, M-SBL~\cite{wipf2007empiricalMSBL},  cyclic coordinatewise optimization (CWO)~\cite{fengler2021nonTIT}, and covariance learning-based orthogonal matching pursuit algorithm (CL-OMP)~\cite{ollila2024matching}.
Joint data and activity detection has also been studied in recent works~\cite{zhang2023jointBiGAMP, wang2021efficientJoint}. 
The AD in different mMTC systems, including multicell and cell-free setups, has also been investigated~\cite{wang24MultiCellAD, zhang24ADCELLFREESBL, rajoriya2024novelCELLFREE}. 
Additionally, model-based learning methods have been proposed based on unrolling classical algorithms like AMP using deep neural networks~\cite{shi2021unrollMassive, zou2024proximal, ma2023modelnoncoherent}.

\revised{
While modern wireless communication systems rely on Gaussian noise assumption with remarkable success, some environments may experience deviations due to outliers and impulsive noise, resulting in heavy-tailed distributions~\cite{zoubir2012robustsurvery, zoubir2018robust, gulgun2023massiveCauchy}.
Sources of such non-Gaussian noise include electromagnetic interference in industrial settings~\cite{gao2007spaceTWC, liu2015impulsiveTVT}, impulsive noise in mobile communication channels~\cite{middleton1999nongaussianTIT}, and heavy-tailed interference in wireless links~\cite{clavier2020experimental}. 
These phenomena represent specific cases where additional robustness can provide benefits.
In mMTC systems deployed in diverse environments, the variability in noise characteristics motivates the development of methods that maintain performance across different conditions~\cite{wang1999robusttsp, jakubisin2016approximate}. 
Traditional methods based on the Gaussian noise assumption may experience performance degradation when faced with heavy-tailed noise or outliers, potentially affecting reliability in IoT networks, for example, where devices operate in varying environmental conditions. 
For additional context on heavy-tailed noise settings in particular massive MIMO scenarios, we refer the reader to~\cite[Section~I-A]{gulgun2023massiveCauchy}.
}

To address the limitations of the Gaussianity assumption in handling non-Gaussian environments, we adopt robust statistical techniques that do not depend on this assumption. Specifically, we replace a Gaussian loss function by robust loss function that effectively mitigates the impact of heavy-tailed noise and outliers. 
\revised{These loss functions do not require precise prior knowledge of the noise distribution and they are related to M-estimation of scatter matrix within the framework of complex elliptically symmetric (CES) distributions~\cite{ollila2012surveyelliptical}, where M-estimation is a generalization of maximum likelihood estimation (MLE) robust to deviations from model assumptions.}
Adopting robust loss function significantly improve reliability of AD in scenarios with impulsive or heavy-tailed noises, conditions frequently encountered in wireless communications and mMTC systems~\cite{feintuch2023neural, feintuch2023neuralDOA, chen2024quantheavytailTIT, wang2014connectivityheavytailTWC, xiao2019vehicelheavytailIOTJ, chavali2011maximum}, yet performing similarly in the conventional Gaussian setting.

\subsection{Main Contributions}
Under Gaussian channel and noise, the device activity detection problem can be formulated as an MLE problem, in which the sample covariance matrix (SCM) of received signals is a sufficient statistic~\cite{chen2024CLmethod}. 
The resulting (Gaussian) negative log-likelihood function (LLF) is non-robust as it fails to down-weight data points with large generalized (Mahalanobis) distances. 
To address this, we propose using robust loss function that enhance robustness against non-Gaussian noise. 
However, this approach results in an objective function that lacks a closed-form solution, even in the conditional case, where other parameters are fixed, and the objective is solved for a single remaining parameter only.
\revised{We show that the coordinate-wise (conditional) objective function is geodesically convex and derive an efficient fixed-point (FP) algorithm for minimizing it, along with convergence guarantees.}
Building on the FP algorithm, we develop two robust AD methods: robust coordinatewise optimization (RCWO) algorithm and robust covariance learning-based matching pursuit algorithm (RCL-MP).

The proposed methods are then evaluated in extensive numerical experiments under various synthetic and realistic heavy-tailed noise conditions. 
The results show that the proposed  methods consistently outperform Gaussian-based counterparts in detection performance. 
Furthermore, the proposed RCL-MP algorithm achieves computational efficiency comparable to Gaussian-based CWO, making it practical for large-scale deployments. The RCWO algorithm on the other hand achieves better performance than the greedy RCL-MP algorithm when the number of active devices are large. 

Our proposed work is closely related to but distinct from previous studies. In~\cite{bai2023robustAD}, the authors employ low-dimensional approximation techniques, such as principal component analysis (PCA) of the channel, to enhance activity detection under highly varying channel conditions. In contrast, we focus on environments with non-Gaussian noise and fixed channel.
Effort addressing non-Gaussian noise scenarios has been lacking in mMTC literature and the present work aims to fill this gap.  We note that the authors in~\cite{marata2024activity} proposed a greedy CL-MP algorithm assuming Gaussian data. Our RCL-MP algorithm is a non-trivial extension of this work, making the method robust to outliers and severe noise conditions.  

\subsection{Paper Organization and Notations}
The remainder of this paper is structured as follows.
Section~\ref{sec:systemmodel} establishes the system model.
Section~\ref{sec:method_robustobj} introduces objectives for robust AD.
Section~\ref{sec:alg} gives the derivation of the proposed algorithms.
Section~\ref{sec:simu} provides numerical experiments to show the robust and superior performance of the proposed methods in comparison with six sparse signal recovery algorithms for AD.
Section~\ref{sec:Conclusion} concludes the paper.

Boldface lower case letters such as $\bba$ represent column vectors, while boldface capital letters like $\bbA$ denote matrices. 
The complex number set is denoted as $\mbC$.
The operations $(\cdot)^\top$ and $(\cdot)^\herm$ denote the transpose and Hermitian transpose, respectively.
The identity matrix is represented as $\bbI$.
We use $|\bbA|$ to denote the determinant of matrix $\bbA$, and use $\mathsf{E}[\bbA]$ to denote the expectation of $\bbA$.

\section{System Model}
\label{sec:systemmodel}
We consider an uplink single-cell massive random access scenario with $N$ single-antenna machine type devices (MTDs) communicating with a base station (BS) equipped with~$M$ antennas.
For the purpose of device identification, each device $n$ is preassigned a unique signature sequence $\bba_n = (a_{n1}, \ldots, a_{nL})$ of length $L$, which is known at the BS.
We also assume that the user traffic is sporadic i.e., only~$K\ll N$ devices are active during each coherence interval. 
The objective is to detect which subset of devices is active based on the received signal at the BS. 
In massive MTC scenarios, also the number of devices is greater than the pilot length (i.e.,~$N > L$ or $N \gg L$).

The generative signal model can be written as
\begin{align}
    \label{eq:basicCSImodel}
    \bbY = \bbA \bbX + \bbE, 
\end{align}
where $\bbY = (\bby_{1} \ \cdots \ \bby_{M}) \in \mbC^{L \times M}$ denotes the received signal matrix over $L$ signal dimensions (symbols) and $M$ antennas, the columns of $\bbY$ are independent due to independent and identically distributed (i.i.d.) channel coefficients over different antennas, $\bbA = (\bba_{1} \ \cdots \ \bba_{N}) \in \mbC^{L\times N}$ is the sensing matrix collecting pilot sequences of devices (known over-complete pilot matrix),
$\bbX = (\bbx_{1}  \ \cdots \ \bbx_{N})^\top \in \mbC^{N \times M}$ is the unobserved channel matrix with
\begin{align}
\bbx_{i} = \sqrt{\gamma_{i}} \bbh_{i}, \ \text{for} \ i=1,\ldots, N
\end{align}
modeling the channel vector between the $i$-th device and the BS, where $\bbh_{i} \sim \mbC \ccalN_M(\bbzero, \bbI)$ is the Rayleigh fading component, and $\gamma_{i}$ is the scaling component defined as
\begin{align}
\gamma_{i} = \alpha_{i} \varrho_{i} \beta_{i},
\end{align}
where~$\alpha_i \in \{0, 1\}$~is an indicator function of device activity ($= 1$ when device is active and $= 0$ otherwise), $\varrho_{i}$ is the device's uplink transmission power, and $\beta_{i} > 0$ is the large-scale fading component (LSFC) accounting for path-loss and shadowing.
Since only $K \ll N$ devices are active, the signal matrix $\bbX$ is $K$-rowsparse, i.e., at most $K$ rows of $\bbX$ contain non-zero entries.
The rowsupport of $\bbX \in \mbC^{N\times M}$ is the index set of rows containing non-zero elements
\begin{align}
    \ccalM = \textrm{supp}(\bbX) = \{i\in [\! [  N ] \!]    : x_{ij} \not = 0 \textrm{ for some }  j \in [\! [   M ] \!]   \}.
\end{align}
Thus, $\ccalM$ collects the indices of the active devices, $\ccalM = \{i \in \llbracket N \rrbracket: \alpha_i = 1\}$.
The objective of activity detection is thus to identify the support set $\ccalM$, given the received signal~$\bbY$, the pilot matrix~$\bbA$, and the noise level~$\sigma^2$.

Further assuming that the noise elements are independent and identically circular Gaussian distributed, $e_{lm} \sim \mathcal C \mathcal N(0,\sigma^2)$ with known variance $\sigma^2>0$, one has that 
$\bby_{m} \sim \mathcal C \mathcal N_L(\mathbf{0},\bbSigma)$ are i.i.d with $L \times L$ positive definite Hermitian (PDH) covariance matrix $\bbSigma = \cov(\bby_m)$ given by~\cite{liu2018sparse}
 \begin{align} \label{eq:bbSigma}
    \bbSigma = \bbA\bbGamma\bbA^{\herm} + \sigma^2\bbI 
    = \sum_{i=1}^N\gamma_i\bba_i\bba_i^{\herm} + \sigma^2\bbI,
\end{align}
where $\bbGamma = \diag(\bbgamma)$ and $\bbgamma = (\gamma_1,\ldots,\gamma_N)^\top$. 

Since $\ccalM = \textrm{supp}(\bbX) = \textrm{supp}(\bbgamma)$, \textit{covariance learning} (CL)-based support recovery algorithms can be constructed by minimizing the Gaussian negative log-likelihood function (LLF) of the data $\bbY$ defined by (after scaling by $1/M$ and ignoring additive constants)
\begin{align}
    \label{eq:elliptical_loss_pre}
    \ccalL(\bbgamma) 
    & = \frac1M \sum_{m=1}^M \bby_m^{\herm}\bbSigma^{-1}\bby_m  +\log|\bbSigma|  .
\end{align}

CL-based methods treat $\bbgamma$ as a set of deterministic but unknown parameters and model the received signal $\bbY$ based on the Gaussianity of both the source signal $\bbX$ and the noise $\bbE$. 
This distribution assumption formulates the problem as a maximum likelihood estimation (MLE) problem, making it tractable for analysis and implementation.
Despite the inherent non-convexity of the optimization problem in~\eqref{eq:elliptical_loss_pre}, which arises from the convexity of $\bby_m^{\herm}\bbSigma^{-1}\bby_m$ and the concavity of~$\log|\bbSigma|$, several efficient CL methods have been developed to solve it, such as SBL~\cite{tipping2001sparsetruSBL}, M-SBL~\cite{wipf2007empiricalMSBL}, variants of CWOs~\cite{haghighatshoar2018improved, fengler2021nonTIT, chen2019covarianceICC}, and greedy methods~\cite{ollila2024matching,marata2024activity}.
Especially, component-wise algorithms~\cite{haghighatshoar2018improved, fengler2021nonTIT, chen2019covarianceICC, ollila2024matching} iteratively update the variable associated with each device, providing closed-form solutions for subproblems.
This reduces the computational complexity and enhances the algorithm efficiency~\cite{chen2024CLmethod}.
Additionally, CL algorithms have better performance with a large number of antennas~$M$ and a smaller pilot length~$L$ compared to AMP methods in low-latency mMTC scenarios~\cite{chen2019covarianceICC}.

In the aforementioned methods, Gaussianity is widely assumed. 
However, real-world scenarios frequently deviate from Gaussian assumptions due to the presence of outliers (resulting in heavy-tailed distributions of received signals), imperfect or outdated channel state information or changes in device mobility and environmental conditions. In such cases, the traditional Gaussian assumption becomes invalid, leading to potential issues in estimation accuracy. 
Consequently, robust methods are required to effectively address situations where the Gaussianity is not valid.

\section{Robust objective function for activity detection} 
\label{sec:method_robustobj}
Our goal is robust detection of active devices.  
Non-robustness of the solution to Gaussian negative LLF~\eqref{eq:elliptical_loss_pre} stems from the property that 
\revised{large (squared) Mahalanobis distances $\{\bby_m^{\herm}\bbSigma^{-1}\bby_m\}_{m=1}^M$ are not down-weighted} 
and can thus have an unduly large influence in the obtained solution. 
A commonly used approach in robust statistics~\cite{maronna2006robust,zoubir2018robust} to reduce the impact of such outlying observations is to downweight them via robust loss function, which is defined as follows. First, recall that a  function  $\rho(t)$ is said to be geodesically convex in $t \in \reals_{\geq 0} $ if $r(x)=\rho(e^x)$ is convex in $x \in \mathbb{R}$. 

\begin{definition}[loss function]\label{def:robustloss}
A geodesically convex function  $\rho(t)$ that is  continuous in $t \in \reals_{>0}$, non-decreasing and differentiable,  is called a loss function,  
and its first derivative \begin{align}
    \label{eq:def_weight}
    u(t) = \rho'(t) \geq 0 
\end{align}
\revised{is called a weight function.}
\end{definition}

\noindent Then, a more general version of \eqref{eq:elliptical_loss_pre} is defined by 
\begin{align} 
    \label{eq:elliptical_loss}
    \ccalL(\bbgamma) 
    & = \frac1{bM}\sum_{m=1}^M \rho\left(\bby_m^{\herm}\bbSigma^{-1}\bby_m\right) +\log|\bbSigma|, 
\end{align}
where $\rho$ is  a loss function in sense of Definition~\ref{def:robustloss}. 
\revised{This robust loss function generalizes the Gaussian MLE in~\eqref{eq:elliptical_loss_pre} by applying~$\rho$ to the Mahalanobis distance, reducing the impact of outliers and handling non-Gaussian noise effectively.
}
The term $b$ is a consistency factor defined as 
\begin{align} 
    \label{eq:consitency factor}
    b = \mathsf{E}[\psi( \bby^\herm \bbSigma^{-1} \bby )]/L, \
    \bby \sim \mathbb{C}\mathcal N_L(\mathbf{0},\bbSigma),   
\end{align} 
where $\psi(t)=t \rho'(t)$.
\revised{This consistency factor is used} in M-estimation for obtaining consistency of the obtained estimator to covariance matrix when data is Gaussian-distributed~\cite{maronna2006robust,zoubir2018robust}. 

There are several loss functions $\rho$ that can be used in~\eqref{eq:elliptical_loss} such as 
 \emph{Gaussian loss} $\rho(t)=t$. Its weight function is ~$u(t)=1$, and consistency factor is $b=1$. Hence  the negative LLF~\eqref{eq:elliptical_loss} reduces to~\eqref{eq:elliptical_loss_pre}.  Gaussian loss is however non-robust. A robust loss function has the property that the weight function descends to zero as $t$ increases. Below we discuss two popular examples of robust loss functions that can be used instead of the Gaussian loss.  

\emph{1). Huber's loss} \cite{ollila2003robust}  is based on a weight function 
\begin{align}
    \label{eq:huberloss}
  u(t;c) 
  = \begin{cases}  
    1, &  \ \mbox{for} \ t \leq c^2, \\ 
    c^2/t,  & \ \mbox{for} \ t > c^2,
    \end{cases}   
\end{align}
where $c>0$ is a  tuning constant that controls robustness (how heavily one down weights large distances). The loss function $\rho$ corresponding to Huber's M-estimator is \cite{ollila2016simultaneous}:
\begin{equation} \label{eq:loss-H} 
\rho(t;c) =
\begin{cases} 
t &   \ \mbox{for} \ t \leqslant c^2, \\ 
c^2 \big( \log (t/c^2) + 1 \big)  & \ \mbox{for} \ t > c^2. 
\end{cases}
\end{equation} 
As~$t= \bby^\herm \bbSigma^{-1} \bby  \sim (1/2)   \chi^2_{2L} $ when~$\bby \sim \mathcal C\mathcal N_L(\mathbf{0},\bbSigma)$, it is common to choose $c^2$
as the $q$th quantile of~$(1/2) \chi^2_{2L}$-distribution, i.e., verifying $2 c^2 = F^{-1}(q; \chi^2_{2L})$ for some~$q \in (0,1)$, where~$F( \cdot ; \chi^2_{2L})$ designates the cumulative density function (cdf) of $\chi^2_{2L}$-distribution. 
We regard $q \in (0,1)$ as a loss parameter which can be chosen by design, and use $q=0.9$ as our default choice.
For $q\to 1$, Huber's loss equals Gaussian loss. 
The scaling constant $b>0$ can be expressed in closed form as
\begin{equation} \label{eq:b_huber}
    b = F_{\chi^2_{2(L+1)}}(2c^2) + c^2(1- F_{\chi^2_{2 L}}(2 c^2))/L.
\end{equation} 
%

\emph{2). t-loss}~\cite[Sec.~4.4.1]{zoubir2018robust} is defined as 
\begin{equation}  \label{eq:loss-T}
\rho(t;\nu)=  \frac{\nu + 2L}{2} \log (\nu +2 t).
\end{equation} 
It originates as being the ML-loss function for the circular complex $L$-variate $t$-distribution (MVT) with $\nu$ degrees of freedom.   
The $\nu$ parameter in~\eqref{eq:loss-T} is viewed as a loss parameter which can be chosen by design.
The case~$\nu=1$ corresponds to the Cauchy loss, while the case~$\nu\to \infty$ yields the Gaussian loss.
A large $\nu$ indicates a strong degree of belief in the Gaussianity assumption.
The respective weight function is
\begin{align}
    u(t;\nu)=   \frac{\nu + 2L}{\nu +2 t}.
\end{align}
In this case, the consistency factor $b$ in~\eqref{eq:consitency factor} is not as easy to derive in closed form, but can always be computed using numerical integration.

\subsection{Solving the Conditional Objective}

In developing a practical algorithm for computing the signal powers, we first consider the simpler case, the conditional objective function where all device powers $\{\gamma_j\}$ for $j\not =i$ are known.  The conditional objective function for  the single unknown $i$-th source power $\gamma$ is
\begin{align}\label{eq:cond_nllf} 
    \ccalL_i(\gamma \mid \bSb ) = 
    \frac{1}{bM}\sum_{m=1}^M \rho(s_{m,i}(\gamma))
    + \log|\bSb + \gamma\bba_i\bba_i^{\herm}|,
\end{align}
where 
\begin{align}
    \label{eq:def_bSb}
    \bSb 
    = \sum_{j\not=i}\gamma_j\bba_j\bba_j^{\herm} + \sigma^2\bbI 
    = \bbSigma - \gamma_i\bba_i\bba_i^{\herm}
\end{align}
is the covariance matrix of $\bby_m$-s after the contribution from the $i$-th device is removed, while 
\begin{align}
    \label{eq:eta_def}
    s_{m,i}(\gamma) 
    = \bby_m^{\herm}\bbSigma^{-1}\bby_m 
    = \bby_m^{\herm}(\bSb + \gamma\bba_i\bba_i^{\herm})^{-1}\bby_m
\end{align}
\revised{is the squared Mahalanobis distance of $\bby_m$.}
Denoting 
\begin{align}
    \bbb_i = \bSb^{-1}\bba_i,  \ d_i =\bba_i^{\herm} \bbb_i, \ i=1,\ldots,N,
\end{align}
and applying matrix inversion lemma\footnote{The matrix inversion lemma (Sherman-Morrison formula) for a rank-one update states that for an invertible matrix $\bbA$ and vectors $\bbu,\bbv$, the inverse of $\bbA + \bbu\bbv^{\herm}$ is given by $ (\bbA + \bbu\bbv^{\herm})^{-1} = \bbA^{-1} - \frac{\bbA^{-1}\bbu\bbv^{\herm}\bbA^{-1}}{1+\bbv^{\herm}\bbA^{-1}\bbu}$.}~to $\bSb + \gamma\bba_i\bba_i^{\herm}$, 
we may write \eqref{eq:eta_def} compactly as 
\begin{align}
    \label{eq:fp_alg_1eta}
    s_{m,i}(\gamma) 
    & = \bby_m^{\herm}  \bSb^{-1}  \bby_m -  \frac{\gamma}{ 1+\gamma d_i} |\bby_m^{\herm}\bbb_i |^2.
\end{align}

Setting the first derivative of $\ccalL_i$ with respect to $\gamma$ to zero yields the estimating equation
\begin{equation} 
    \label{eq:cond_nllf_1stderivative}
    0 = \frac{1}{bM}\sum_{m=1}^M  u_{m,i}(\gamma)   \frac{\partial s_{m,i}(\gamma)}{\partial\gamma}
    + \frac{\partial}{\partial \gamma} \log|\bSb + \gamma\bba_i\bba_i^{\herm}|,
\end{equation}
where we use a shorthand notation $u_{m,i}(\gamma)= u(s_{m,i}(\gamma))$. 
Noting that 
\begin{align}
    \label{eq:cond_nllf_1stderivative_1log_tgamma}
    \frac{\partial s_{m,i}(\gamma)}{\partial\gamma} 
    & = -   |\bby_m^{\herm}\bbb_i |^2 \frac{\partial}{\partial \gamma} \frac{\gamma} {1+\gamma d_i}
    = - \frac{ |\bby_m^{\herm} \bbb_i |^2}{(1+\gamma d_i)^2}, 
\end{align}
the first term on the right hand side (RHS) of~\eqref{eq:cond_nllf_1stderivative} becomes 
\begin{align}
    & \frac{1}{bM} \sum_{m=1}^M u_{m,i}(\gamma)  \frac{\partial s_{m,i}(\gamma)}{\partial\gamma} 
  =   -\dfrac{  \hat{\sigma}^2_{\gamma,i} }{(1+\gamma d_i)^2},\ \ \ 
    \label{eq:cond_nllf_partial3}
\end{align}
 where
\begin{align}
    \label{eq:skewed_SCM}
    \hat{\sigma}^2_{\gamma,i} = \frac{1}{bM} \sum_{m=1}^M  u_{m,i}(\gamma)  |\bby_m^{\herm}\bbb_i |^2. 
\end{align}
For the second term on the RHS of~\eqref{eq:cond_nllf_1stderivative}, we have
\begin{align}
    \label{eq:cond_nllf_partial2_sub1}
    \frac{\partial \log|\bSb + \gamma\bba_i\bba_i^{\herm}| }{\partial \gamma}
    & = \frac{ d_i}{1+\gamma d_i}.
\end{align}
Combining \eqref{eq:cond_nllf_partial3} and \eqref{eq:cond_nllf_partial2_sub1} in \eqref{eq:cond_nllf_1stderivative} yields
\begin{align}
    \label{eq:indexelliptical_nllf_1stderivative}
    -\dfrac{  \hat{\sigma}^2_{\gamma,i}  }{(1+\gamma d_i)^2}
    + \frac{ d_i}{1+\gamma d_i} = 0,
\end{align}
which shows that a minimizer of \eqref{eq:cond_nllf} must verify a FP equation of the form
\begin{align}\label{eq:fp_gamma_good}
    \gamma      
    = \ccalH_i(\gamma)
    = \frac{\hat{\sigma}^2_{\gamma,i}  - d_i}{ d_i^2}.
\end{align}

\subsection{FP Algorithm}
\label{subsec:fp_proposal}

In this subsection, we detail the proposed FP algorithm for solving the minimizer of $\ccalL_i(\gamma) \equiv \ccalL_i(\gamma \mid \bSb)$
for each coordinate/device $i=1,\ldots,N$.
We derive a FP algorithm using \eqref{eq:fp_gamma_good} that proceeds as follows.
Start with an initial guess $\gamma^{(0)} \geq 0$ and perform the iterative procedure
\begin{align}
    \label{eq:FP_mapping_iteration}
    \gamma^{(\ell+1)} = \ccalH_i(\gamma^{(\ell)}),
\end{align}
which implies computing the following steps
\begin{subequations}
    \begin{align}
        s_{m,i}^{(\ell)} 
        & = \bby_m^{\herm}  \bSb^{-1}  \bby_m -  \frac{\gamma^{(\ell)}}{ 1+\gamma^{(\ell)}  d_i} |\bby_m^{\herm}\bbb_i |^2, \label{eq_fp_alg:1_updateSigma}  \\   
         \hat{\sigma}^{2(\ell)}_{\gamma,i} &= \frac{1}{bM} \sum_{m=1}^M  u( s_{m,i}^{(\ell)})  |\bby_m^{\herm}\bbb_i |^2, \label{eq_fp_alg:4_nonnegative_sigma} \\ 
        \gamma^{(\ell+1)} &= 
        \frac{  \hat{\sigma}^{2(\ell)}_{\gamma,i} - d_i}{d_i^2} . 
        \label{eq_fp_alg:4_nonnegative_gamma}
    \end{align}
\end{subequations}
The iterative procedure given by \eqref{eq_fp_alg:1_updateSigma}-\eqref{eq_fp_alg:4_nonnegative_gamma} defines a sequence $\{\gamma^{(\ell)} \}$. 
Next, we establish geodesic convexity of the conditional objective function. 
\begin{theorem}[geodesic convexity]
    \label{thm:gconvexity_Li}
    If $\rho(x)$ is a loss function in sense  of Definition~\ref{def:robustloss}, then $\ccalL_i(\gamma)$ is geodesically convex in $\gamma \in\reals_{\geq0}$.
\end{theorem}
\begin{proof}
    See Appendix~\ref{proof:thm_gconvexity_Li}.
\end{proof}
Given the geodesical convexity established in Theorem~\ref{thm:gconvexity_Li}, we now characterize the global minimizer properties of the conditional objective function $\ccalL_i(\gamma)$.

\revised{
\begin{theorem}[existence and uniqueness]
    \label{thm:globalminimum}
    If $\rho(\cdot)$ is a loss function in Definition~\ref{thm:gconvexity_Li}, then:
    \begin{itemize}
        \item[(a)] Any local minimum of $\mathcal{L}_i(\gamma)$ is a global minimum.
        \item[(b)] If $\rho(\cdot)$ is bounded below, then $\mathcal{L}_i(\gamma)$ has at least one global minimizer.
        \item[(c)] If, in addition, $\rho(\cdot)$ is strictly geodesically convex, then the global minimizer of $\mathcal{L}_i(\gamma)$ is unique.
    \end{itemize}
\end{theorem}
\begin{proof}   
    See Appendix~\ref{proof:thm_globalminimum}. 
\end{proof}
}
\revised{
Theorem~\ref{thm:globalminimum} has established the existence and uniqueness properties of the global minimizer for the conditional objective function $\ccalL_i(\gamma)$.}
Next we show that  for any $\gamma^{(\ell)}$ in the sequence $\{\gamma^{(\ell)} \}$  for $\ell=0,1,2,\ldots$, that  is not a stationary point of the objective function $\ccalL_i(\gamma)$, the objective function decreases at successive iterations, i.e., the  
$\{\ccalL_i(\gamma^{(\ell)})\}$  forms a monotone decreasing sequence in $\reals$.

\begin{theorem}[monotonic decrease]
    \label{thm:monotonicityofsequence}
    Assume that $\rho(\cdot)$ has second derivative with $\rho''(\cdot) \leq 0$.
    Let $\gamma^{(\ell+1)} > 0$, and $\gamma^{(\ell)} > 0$ be two successive iterations obtained from the FP mapping~\eqref{eq:FP_mapping_iteration}.
    If $\gamma^{(\ell)}$ is not a stationary point of $\ccalL_i(\gamma)$, then
    \begin{align} \label{eq:proof:inequality}
         \ccalL_i(\gamma^{(\ell+1)}) \leq \ccalL_i(\gamma^{(\ell)}).
    \end{align}
    \revised{
    The inequality in~\eqref{eq:proof:inequality} is strict if $\rho(\cdot)$ is increasing.}
\end{theorem}
\begin{proof}   
    See Appendix~\ref{apdx:monotonicity}. 
\end{proof}
\revised{
Theorem~\ref{thm:monotonicityofsequence} shows that the conditional negative LLF decreases at each step of the FP algorithm and the negative LLF's values $\{\ccalL_i(\gamma^{(\ell})\}_{\ell}$ form a monotonic sequence.}
\begin{itemize}
    \item \revised{Interior case ($\gamma > 0$): 
    The geodesic convexity ensures that any local minimum is a global minimum.
    When $\rho(\cdot)$ is bounded below, the function $\mathcal L_i(\gamma)$ is bounded below and coercive (see Appendix~\ref{proof:thm_globalminimum}), guaranteeing that the FP iterations remain bounded.
    In practice, $\rho(\cdot)$ is selected as an increasing function (see for example \eqref{eq:loss-H} and \eqref{eq:loss-T}).
    Thus, \eqref{eq:proof:inequality} holds as strict inequality and FP iterations converge to the global minimum.
    }
    \item \revised{Boundary case ($\gamma = 0$): Algorithm~\ref{alg:basic_fp} (FP-CW) efficiently identifies this case by comparing $\ccalL_i(0)$ with $\ccalL_i(\epsilon)$ for a small positive $\epsilon$.
    If $\ccalL_i(0) < \ccalL_i(\epsilon)$, the algorithm confirms $\hat{\gamma}=0$ as the minimizer without further iterations.}
\end{itemize}

\begin{algorithm}[!t]
    \caption{\textbf{FP-CW}: \textbf{FP} algorithm for \textbf{c}oordinate-\textbf{w}ise objective function}
    \begin{algorithmic}[1]
    \State\textbf{Input:} 
    $ ( \bby_m^\herm \bbSigma_{\backslash i}^\inv \bby_m)_{M \times 1}$,  $( | \bby_m^\herm \bbb_i|^2)_{M \times 1}$,  $d_i>0$, weight function $u(\cdot)$, initial device power~$\hat \gamma_{\text{init}} \geq 0$ 
    
    \State\textbf{Initialize:} convergence threshold: $\delta = 5 \cdot 10^{-3}$; a number close to $0$: $\epsilon = 10^{-3}$; maximum number of iterations: $I_{\max}=10$; $\hat \gamma =\hat \gamma_{\text{init}}$.

    \If{$\ccalL_i(0) < \ccalL_i( \epsilon)$}
        \State $ \hat \gamma \gets 0$
    \Else
    \Comment{\textrm{Minimum is not on boundary, so use FP algorithm}}
    
    \State $\hat \gamma^{\text{old}}  \gets \hat{\gamma}_{\text{init}}$ 
    
    \For{$ \ell = 1, \ldots, I_{\max}$}

        \State $     s_{m,i} \gets  \bby_m^\herm \bbSigma_{\backslash i}^\inv \bby_m- \dfrac{\hat \gamma}{1+\hat\gamma d_i}  | \bby_m^\herm \bbb_i|^2 $
         \label{alg_tab:step_Mdistance}
        
        \State      $   {\displaystyle  \hat{\sigma}^2_{\gamma,i} \gets \frac{1}{bM} \sum_{m=1}^M u(s_{m,i})   | \bby_m^\herm \bbb_i|^2  } $       \label{alg_tab:step_weightedSCM}
        \State
        $
        \hat \gamma \gets {\displaystyle \frac{\hat{\sigma}^2_{\gamma,i} -  d_i}{ d_i^2} }
        $
        \label{alg_tab:step_devicePower}
        
        \If{$ \ell = I_{\max}$ \textbf{or} $| \hat \gamma - \hat \gamma^{\text{old}}|/ \hat \gamma^{\text{old}} < \delta$}
            \State \textbf{Break}
        \EndIf
        \State  $ \hat \gamma^{\text{old}} \gets  \hat \gamma$
    \EndFor
    \EndIf

    \State \textbf{Return:} optimal $ \hat \gamma \in \mathbb{R}_{\geq 0}$
    \end{algorithmic}
    \label{alg:basic_fp}
\end{algorithm}

Algorithm~\ref{alg:basic_fp} summarizes the FP mapping given by \eqref{eq_fp_alg:1_updateSigma}-\eqref{eq_fp_alg:4_nonnegative_gamma} which fits the coordinate-wise optimization in the AD problem.
The major computational complexity of the FP-CW algorithm lies in the input stage, where the term $\bby_m^\herm \bbSigma_{\backslash i}^\inv \bby_m$ is computed for $m=1,\ldots,M$, resulting in a complexity of~\( \mathcal{O}(ML^2) \).

\section{Algorithms}
\label{sec:alg}
The FP-CW algorithm proposed in Section~\ref{subsec:fp_proposal} serves as the foundation for developing two alternative approaches for resolving the support and the non-zero signal powers.
The first is a coordinatewise optimization (CWO) algorithm while the second is a matching pursuit (MP) algorithm. 

\subsection{Robust Coordinatewise Optimization Algorithm} 

We develop a robust coordinatewise optimization (RCWO) algorithm using covariance learning, \revised{which is similar to the coordinate descent framework} in~\cite{fengler2021nonTIT, haghighatshoar2018improved}. 
This approach updates the parameters $\gamma_1, \ldots, \gamma_N$ cyclically by solving
\begin{align*}
    \min_\gamma \ccalL_i(\gamma)
\end{align*}
for $i=1,\ldots,N$. That is, we estimate powers one coordinate at a time while keeping others fixed at their current iterate values. 
In RCWO, each device power $\gamma_i$ is updated one at a time via the FP-CW algorithm~(Algorithm~\ref{alg:basic_fp}), and afterwards the estimation of covariance matrix $\bbSigma$ is updated by the new~$\gamma_i$.
The algorithm is as follows.

\textbf{Initialization.} Set $\hbgamma = \bbzero$, and $\bbSigma^\inv = \sigma^{-2}\bbI$ as initial solutions of signal powers and the initial covariance matrix at the start of iterations, respectively.
Compute $\bbc_i = \bbSigma^\inv \bba_i$.

\begin{algorithm}[!h]
    \caption{ \textbf{RCWO}: \textbf{R}obust \textbf{c}oordinate\textbf{w}ise \textbf{o}ptimization algorithm}
    \begin{algorithmic}[1]
    \State\textbf{Input:}
    $\bbA$, $\bbY$, $\sigma^2$, $K$, function $u(\cdot)$.
    \State\textbf{Initialize:}
    $ \bbSigma^\inv = \sigma^{-2}\bbI$, $\hbgamma = \bbzero$, number of cycles $I_\textup{cyc}=50$, and stopping threshold for CWO algorithm $\delta_{\textup{CWO}} = 5\cdot 10^{-3}$.
    \For{$i_\textup{cyc}=1,\ldots,I_\textup{cyc}$}
        \For{$i = 1,\ldots, N$}
        \vspace{0.3mm}

        \State $\bbc_i \gets \bbSigma^\inv \bba_i$
        \State $\bbTheta \gets \bbSigma^\inv + \dfrac{\hhatgamma_i \bbc_i \bbc_i^\herm}{1 - \hhatgamma_i \bba_i^\herm \bbc_i}$
        
        \State $\bbB = \begin{pmatrix} \bbb_1  & \cdots  &\bbb_N \end{pmatrix} \gets \bbTheta \bbA$
        
        \State  $ \tilde{\boldsymbol{s}}= (\tilde s_m)$, where  $\tilde s_m \gets \bby_m^{\herm} \bbTheta \bby_m$, $m=1,\ldots,M$.
        
        \State  $\boldsymbol{v} = (v_m)$, where  $v_m \gets  |\bby_m^{\herm}\bbb_i |^2$, $m=1,\ldots,M$.
         
        \State  $d_i  \gets \bba_i^\herm \bbb_i$ 
         
        \State 
        $\hhatgamma_i^{\textup{new}} \gets \textbf{FP-CW}( \tilde{\bbs}, \bbv, d_i, u(\cdot),  \hat \gamma_i)$
    \vspace{0.3mm}

        \State 
        $\delta_i \gets \hhatgamma_i^{\textup{new}} - \hhatgamma_i$
        \vspace{0.3mm}

        \State
        $ \bbSigma^\inv \gets \bbSigma^\inv - \dfrac{\delta_i \bbc_i \bbc_i^\herm}{1+\delta_i \bba_i^{\herm}\bbc_i }$
        \vspace{0.3mm}

        \State
        $\hhatgamma_i \gets \hhatgamma_i^{\textup{new}}$
        \vspace{0.3mm}
        \EndFor

        \If{$\|\hbgamma - \hbgamma^{\textup{new}}\|_\infty/\|\hbgamma^{\textup{new}}\|_\infty < \delta_{\textup{CWO}} $}
            \State \textbf{Break}
        \EndIf
        
    \EndFor
    
    \State $\ccalM \gets \mathsf{supp}_K \left( \hbgamma \right)$
    
    \State\textbf{Return:} Set $\ccalM$ of indices for active devices ($|\mathcal M | = K$)
    \end{algorithmic}
    \label{alg:rcwo}
\end{algorithm}

\textbf{Main iteration.}

{\it 1) Compute $\bSbi$} using the Sherman-Morrison formula
\begin{align}
    \label{eq_cwo:backslashSigmaforFP}
   \bbTheta= \bSbi
    &=
    (\bS - \hhatgamma_i\bba_i\bba_i^{\herm})^\inv 
  = \bS^\inv + \frac{\hhatgamma_i \bbc_i \bbc_i^\herm }{1 - \hhatgamma_i\bba_i^{\herm}\bbc_i}.
\end{align}

{\it 2) Minimize the conditional log-likelihood}  $\ccalL_i( \gamma \mid \bSbi)$ using the  FP-CW algorithm. 
This gives the new signal power~$\gamma_i^{\textup{new}}$ estimate.  

{\it 3) Update the inverse covariance matrix} 
using the Sherman-Morrison formula
\begin{align}
    \label{eq_cwo:covariance_new_inv}
    (\bS^{\textup{new}})^\inv 
    & = 
    (\bS + \delta_i \bba_i\bba_i^\herm)^\inv
        = \bS^{-1} - \frac{\delta_i \bbc_i \bbc_i^\herm}{1+\delta_i\bba_i^{\herm}\bbc_i}.
\end{align}
where  $\delta_i = \gamma_i^{\textup{new}} - \hhatgamma_i$ is the difference between old and new powers of the $i$-th signal 

{\it 4) Cycle through indices and terminate} if the relative error 
\begin{align}
    \label{eq_cwo:stop}
\|\hbgamma - \hbgamma^{\textup{new}}\|_\infty/\|\hbgamma^{\textup{new}}\|_\infty 
\end{align}
falls below  a tolerance threshold  $\delta_{\textup{cwo}}$  or the maximum number of cycles (e.g.,~$I_\textup{cyc} = 100$) has been reached.

The estimated support is found by using 
\begin{align}
    \ccalM \leftarrow \textup{supp}_K\left( \bbgamma^{\textup{new}} \right),
\end{align}
where $\textup{supp}_K(\cdot)$ finds the $K$-largest elements.
Algorithm~\ref{alg:rcwo} summarizes the steps.
\revised{
We also note that our empirical experiments comparing cyclic and randomized coordinate update strategies for RCWO showed nearly identical performance across various sparsity levels and pilot lengths.  
This indicates that the choice of update ordering has minimal impact on the detection accuracy for our specific AD problem.}

The dominant complexity of the RCWO algorithm relies on the FP-CW algorithm in line~11 of Algorithm~\ref{alg:rcwo}, which requires the complexity of $\ccalO(TML^2)$.
Thus, for each cycle, the complexity of RCWO algorithm is $\ccalO(TMNL^2)$, where~$N$ is the total number of devices.
For the Gaussian-based CWO algorithm proposed in~\cite{haghighatshoar2018improved, chen2019covarianceICC,fengler2021nonTIT}, its complexity is $\ccalO(NL^2)$, which is smaller than that of the proposed RCWO algorithm.
However, the increased complexity $\ccalO(TMNL^2)$ remains manageable and can be justified by the improved robustness and performance, which is shown in Section~\ref{sec:simu}.
\revised{Table~\ref{tab:complexity} summarizes the complexity of algorithms discussed above, where $k$ denotes the number of atoms selected in SOMP~\cite{tropp2006algorithms}.}

\begin{table}[!t]
\centering
\revised{
\caption{Computational complexity of AD algorithms.}
\begin{tabular}{||c|c||}
\hline
\textbf{Algorithm} & \textbf{Complexity (per iteration)} \\
\hline
CL-MP \cite{marata2024activity} & $\mathcal{O}(NL^2)$ \\
\hline
CWO \cite{fengler2021nonTIT} & $\mathcal{O}(NL^2)$ \\
\hline
AMP \cite{liu2018massiveAMP} & $\mathcal{O}(N^2M)$ \\
\hline
SBL \cite{wipf2004sparse} & $\mathcal{O}(N^2L)$ \\
\hline
SOMP \cite{tropp2006algorithms} & $\mathcal{O}(L \cdot \max(NM, k^2) + k^3)$ \\
\hline
RCWO (proposed) & $\mathcal{O}(TMNL^2)$ \\
\hline
RCL-MP (proposed) & $\mathcal{O}(TMNL^2)$ \\
\hline
\end{tabular}
\label{tab:complexity}
}
\end{table}

\subsection{Robust CL-based Matching Pursuit Algorithm}
\begin{algorithm}[!t]
    \caption{\textbf{RCL-MP}: \textbf{R}obust \textbf{CL}-based \textbf{m}atching \textbf{p}ursuit algorithm}
    \begin{algorithmic}[1]
    \State\textbf{Input:}
    $\bbA$, $\bbY$, $\sigma^2$, $K$, and function $u(\cdot)$. 
    \State\textbf{Initialize:}
    $\bbTheta = \sigma^{-2}\bbI$, $\ccalM = \emptyset$, $\hat \bbgamma = \bbzero$.
    
    \For{$k=1,\ldots,K$}
        \State $\bbB = \begin{pmatrix} \bbb_1  & \cdots  &\bbb_N \end{pmatrix} \gets \bbTheta \bbA$
        \State  $ \tilde{\boldsymbol{s}}= (\tilde s_m)$, where  $\tilde s_m \gets \bby_m^{\herm} \bbTheta \bby_m$, $m=1,\ldots,M$.

        \For{$ i \in \ccalM^\complement $}
           	\State  $\boldsymbol{v} = (v_m)$, where  $v_m \gets  |\bby_m^{\herm}\bbb_i |^2$, $m=1,\ldots,M$.
         \State  $d_i  \gets \bba_i^\herm \bbb_i$ 

            \State 
            $\hat \gamma_i \gets \textbf{FP-CW}( \tilde{\boldsymbol{s}}, \boldsymbol{v}, d_i, u(\cdot),  \hat \gamma_i)$
            \label{alg_tab_rclmp:step_FPCW}
            \vspace{0.3mm}

            \State 
            $s_{m,i} \gets  \tilde s_m - \dfrac{  \hat \gamma_i } {1+ \hat \gamma_i d_i} v_m, \ m=1,\ldots,M$

            \vspace{0.3mm}       
            
            \State
            $\epsilon_i  \gets \frac1{bM}\sum_{m=1}^M \rho (s_{m,i}) - \log|  \bbTheta | + \log (1+ \hat \gamma_i d_i)$
        \EndFor

    \State  $\ccalM \gets \ccalM \cup\{i_k\}$ with $i_k \gets \argmin_{i \in \ccalM^\complement} \epsilon_i$
    \vspace{0.3mm}

    \State $\bbTheta \gets \bbTheta - \dfrac{ \hat \gamma_{i_k}}{1 + \hat \gamma_{i_k} d_{i_k}} \bbb_{i_k}\bbb_{i_k}^{\herm}$
    
    \EndFor
    
    \State\textbf{Return:} Set $\ccalM$ of indices for active devices ($|\mathcal M | = K$)
    \end{algorithmic}
    \label{alg:greedymain}
\end{algorithm}

We develop a Robust CL-based Matching Pursuit (RCL-MP) algorithm, following the generic matching pursuit strategy outlined in~\cite{elad2010sparse}.
Unlike the RCWO updating cyclically, RCL-MP selects the most influential signal index $i$ which gives the minimum $\epsilon_i$ at each iteration and adds it to the support set.
This method updates the parameters by solving the optimization problem
\begin{align}
\label{eq:epsilon_nLLF}
\epsilon_i = \min_{\gamma \geq 0} \ccalL_i(\gamma)
\end{align}
for atoms $i$ in the complement of the current support set $\ccalM^{(k)}$, and adds the atom $i$ to this that minimum value of error $\epsilon_i$ (i.e.,  $i = \argmin_{j} \epsilon_j$).

\revised{ 

{\bf Initialization.} We begin with the following settings. 
\begin{itemize}
    \item 
    The iteration index is set to $k = 0$ and the initial power vector is set to $\boldsymbol{\gamma}^{(0)} = \mathbf{0}_{N \times 1}$, which means that no devices are initially active.
    \item 
    The support set $\mathcal{M}^{(k)}$, which contains the indices of devices identified as active at iteration $k$, is initialized as $\mathcal{M}^{(0)} = \mathsf{supp}(\boldsymbol{\gamma}^{(0)}) = \emptyset$.
    \item The initial covariance matrix is set to $\boldsymbol{\Sigma}^{(0)} =  \bbA \diag(\boldsymbol{\gamma}^{(0)}) \bbA^\herm + \sigma^{2} \bbI =  \sigma^{2} \bbI$, and consequently, the inverse covariance matrix is set to $ \bbTheta^{(0)}= (\bbSigma^{(0)})^{-1}= (1/\sigma^{2}) \bbI $.
\end{itemize}

\textbf{Main iteration.} At each iteration $k = 1, 2, \ldots, K$:

{\it 1) Compute}
$\bbB = \boldsymbol{\Theta}^{(k-1)}\bbA$ and the squared Mahalanobis distances $\tilde{s}_m = \bby_m^{\herm} \boldsymbol{\Theta}^{(k-1)}\bby_m$ for $m = 1, \ldots, M$. 
These values are then reused throughout the iteration, reducing the computational burden.

{\it 2)  Sweep.} 
For each device $i$ not yet in the support set, i.e., $i \in (\mathcal{M}^{(k-1)})^\complement = [\![ N ]\!] \setminus \ccalM^{(k-1)}$:
\begin{itemize}
    \item Compute the elements $v_m = |\bby_m^\herm \bbb_i|^2$ for $m = 1, \ldots, M$, where $\bbb_i$ is the $i$-th column of $\bbB$.
    \item Compute $d_i = \bba_i^\herm \bbb_i$, which represents the squared norm of $\bba_i$ under the metric defined by $\boldsymbol{\Theta}^{(k-1)}$.
    \item Determine the optimal power $\hat{\gamma}_i$ using the FP-CW algorithm.
    \item Compute the updated Mahalanobis distances $s_{m,i} = \tilde{s}_m - \frac{\hat{\gamma}_i}{1 + \hat{\gamma}_i d_i}v_m$ for $m = 1, \ldots, M$.
    \item Compute the associated error $\epsilon_i$
    \begin{align}
        \epsilon_i &=  \min_{\gamma \geq 0} \ccalL_i(\gamma \mid \bbSigma^{(k)} )  \notag  \\
        & = \frac1{bM} \sum_{m=1}^M \rho (s_{m,i}) + \log|  \bbSigma^{(k)} +  \hat \gamma_i  \bba_i\bba_i^{\herm} |. \label{eq:log_det}
    \end{align}
    Note that the log-determinant in \eqref{eq:log_det} can be computed efficiently using formula
    \begin{align}
     &\log \Big|  \bbSigma^{(k)} +   \hat \gamma_i  \bba_i\bba_i^{\herm} \Big|  = \log | \bbSigma^{(k)} ( \boldsymbol{I} + \hat{\gamma}_i  \bbTheta^{(k)} \bba_i \bba_i^\herm) |  \notag  \\ 
     &\quad =  \log | \bbSigma^{(k)}| +   \log (1+ \hat \gamma_i d_i). \label{eq:log_det2}
     \end{align}
     The last identity in \eqref{eq:log_det2} follows from Sylvester's determinant theorem\footnote{It states that for $\bbA$, an $m \times n$ matrix, and $\bbB$, an $n \times m$ matrix, it holds that  $| \bbI _m + \bbA \bbB | = | \bbI_n + \bbB \bbA |$.}.
    Using the fact that $\bbTheta^{(k)}$ is the inverse of $\bbSigma^{(k)}$, and thus $\log | \bbSigma^{(k)} | = - \log | \bbTheta^{(k)} |$, the error $\epsilon_i$ can be computed as
    \begin{equation}
        \epsilon_i = \frac1{bM}\sum_{m=1}^M \rho (s_{m,i}) - \log|  \bbTheta^{(k)} |  + \log (1+ \hat \gamma_i d_i). \nonumber 
    \end{equation}
\end{itemize}

{\it 3) Update support.} 
Select the device index $i_k = \arg\min_{i\in(\mathcal{M}^{(k-1)})^\complement} \epsilon_i$ that gives the minimum error and add it to the support set: $\mathcal{M}^{(k)} = \mathcal{M}^{(k-1)} \cup {i_k}$.

{\it 4) Update the inverse covariance matrix.} The inverse covariance matrix can be efficiently updated using the Sherman-Morrison formula
\begin{align}
\bbTheta^{(k+1)} = \bbTheta^{(k)} - \dfrac{ \hat \gamma_{i_k}}{1 +  \hat \gamma_{i_k} d_{i_k}} \bbb_{i_k}\bbb_{i_k}^{\herm}. \nonumber
\end{align}

{\bf  Stopping criterion.} This process is repeated until $K$ devices are identified, at which point the algorithm terminates and returns the final support set $\mathcal{M}^{(K)}$ as the set of active devices.

Algorithm~\ref{alg:greedymain} displays the pseudo-code of RCL-MP method.
Note that, consistent with common practice in MP approaches in the literature~\cite{tropp2006algorithms}, the RCL-MP algorithm assumes prior knowledge of the number of active devices $K$. 
This assumption is standard across all covariance-based activity detection methods.
}

The dominant complexity of proposed RCL-MP algorithm relies on the FP-CW algorithm in line~\ref{alg_tab_rclmp:step_FPCW} of Algorithm~\ref{alg:greedymain}, which requires the complexity of $\ccalO(TML^2)$.
For each increment in $k$ by $1$, the sweeping range $|(\ccalM^{(k-1)})^\complement|$ decreases by $1$, where $|(\ccalM^{(0)})^\complement|=N$ for $k=1$.
Thus, for each iteration of RCL-MP algorithm, the worst-case complexity is $\ccalO(TMNL^2)$.
The computational complexity of RCLMP algorithm (Algorithm~\ref{alg:greedymain}) is the same as for RCWO algorithm (Algorithm~\ref{alg:rcwo}).
Note that, in the empirical executions, the FP algorithm usually converges within $T=5$ steps.
RCWO method can run for maximum $I_\textup{cyc}$ rounds, while RCL-MP terminates after $K$ rounds. In mMTC scenarios, it is common that $K<I_{\textup{cyc}}$. 
As for other methods widely used in the AD problem, the SBL~\cite{wipf2004sparse} has complexity of $\ccalO(N^2L)$, the VAMP~\cite{liu2018massiveAMP} has complexity of $\ccalO(N^2M)$, and the SOMP~\cite{tropp2006algorithms} has complexity of $\ccalO(L\cdot\max(NM,k^2)+k^3)$, where $k$ is the index of current iteration.

The complexity of our proposed methods (RCWO and RCL-MP) scale with the number of antenna elements $M$.  
This is not desirable in the massive-MIMO/large-antenna system, where $M$ can be very large.
The computational complexity of Gaussian-based methods,  CWO ~\cite{haghighatshoar2018improved,fengler2021nonTIT, chen2024CLmethod} and CL-MP \cite{marata2024activity} 
is $\ccalO(NL^2)$, which is more suitable in massive MIMO scenarios. 

This proposed RCL-MP algorithm reduces to CL-MP algorithm~\cite{marata2024activity} if one uses Gaussian loss ($\rho(t)=t$). 
For Gaussian loss, the device's power $\gamma_i$ has a closed form solution~\cite{marata2024activity} given by
\begin{align}
   \hat \gamma_i = \max \left(
    \frac{\bba_i^{\herm} \bSb^{-1}   ( \frac1M \bbY\bbY^\herm - \bSb)  \bSb^\inv\bba_i}{(\bba_i^{\herm}\bSb^{-1}\bba_i)^2}, 0 \right).
\end{align}
In this case, $u(t)=1, \forall t$, so adaptive weighting in \eqref{eq:skewed_SCM} is not necessary. 
Thus, the CL-MP benefits from a lower computational complexity of $\ccalO(NL^2)$.
However, when non-Gaussian noise models are used, such as Huber's loss or $t$-loss, the assumption that all observations have the same importance is no longer valid. 
These robust losses deal with heavy-tailed distributions or impulsive noise by down-weighting the influence of outliers in the data, where the outliers are featured by large squared Mahalanobis distances~\eqref{eq:eta_def}.
The adaptivity of~\eqref{eq:skewed_SCM} adjusts the contribution of each observation and thus provides robustness.
However, the increased complexity $\ccalO(TMNL^2)$ is the cost for additional robustness.
Also note that when the quantile parameter $q$ of Huber's loss~\eqref{eq:huberloss} is  $q \approx 1$ (so tuning constant $c^2$ is large), we expect CL-MP and RCL-MP to perform similarly.

\section{Numerical Experiment}
\label{sec:simu}

This section validates the performance of our proposed RCWO and RCL-MP algorithms, comparing them with common algorithms for the AD problem under a variety of settings.
We use the following performance metrics for this evaluation: the (empirical) probability of exact recovery ($\per$), the probability of missed detection ($\pmd$), and the probability of false alarm ($\pfa$).
Exact recovery occurs when all active devices are correctly detected. We define the (empirical) probability of exact recovery ($\per$)~as 
\begin{align} \per = \frac{1}{N_{\textup{sim}}} \sum_{n=1}^{N_{\textup{sim}}} \mathds{1} (\ccalM_{n} = \ccalM^\circ), 
\end{align} 
where $N_{\textup{sim}}$ is the number of Monte Carlo trials, $\mathds{1}(\cdot)$ denotes the indicator function, $\ccalM^\circ$ is the true set of active devices, and $\ccalM_{n}$ is the estimated support for the $n$-th Monte Carlo trial.
Missed detection occurs when an active device is incorrectly identified as inactive. The (empirical) probability of missed detection ($\pmd$) is  defined the ratio of  missed devices relative to the total number of active devices 
averaged over MC trials:
\begin{align} \pmd & = \frac{1}{N_{\textup{sim}}} \sum_{n=1}^{N_{\textup{sim}}} \dfrac{|\ccalM^\circ \setminus \ccalM_{n} |}{|\ccalM^\circ|}. 
\end{align}
\revised{
Since we select exactly $K$ devices as active in our simulations (where $K = |\mathcal{M}^\circ|$ is the true number of active devices), the number of false alarms equals the number of misdetections, 
thus, $\pfa =\pmd$ (see also~\cite{chen2018sparse}).
Consequently, in the following experiments, we only show $P_{\text{MD}}$.}
All experiments are conducted on a workstation equipped with a 16-core Intel Xeon w5-2465X CPU and 192 GB of memory.
The code is available at \url{https://github.com/xnnjw/RobustAD_release}.

\subsection{Synthetic Simulation}
\label{subsec:synthetic_simulation}
\begin{figure}[!t]
    \centerline{\includegraphics[width = 0.99\linewidth]{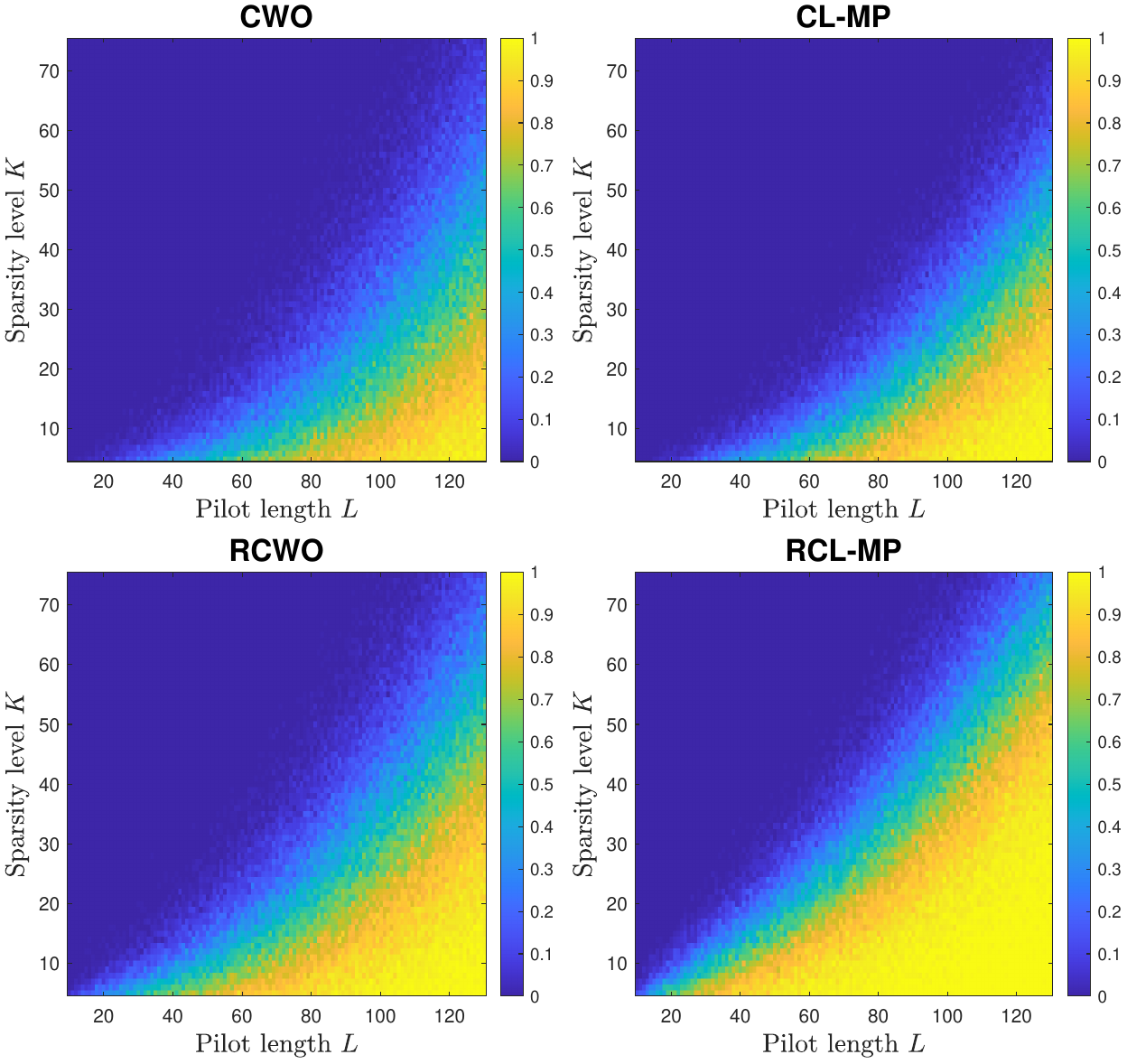}}
    \caption{
        \revised{Phase transition of the average $\per$ for the CWO, CL-MP, RCWO, and RCL-MP algorithms, with $N=1000$ and $M=100$. 
        The noise follows an $\epsilon$-contaminated Gaussian distribution ($\epsilon=0.05$, $\lambda=10$).
        The tuning parameters of Huber loss function in RCWO and RCL-MP are set as $q = 0.9$.
        } 
    }
    \vspace{-3mm}
    \label{fig:phase_transition}
\end{figure}
\begin{figure}[!t]
    \centering
    \includegraphics[width=0.99\linewidth]{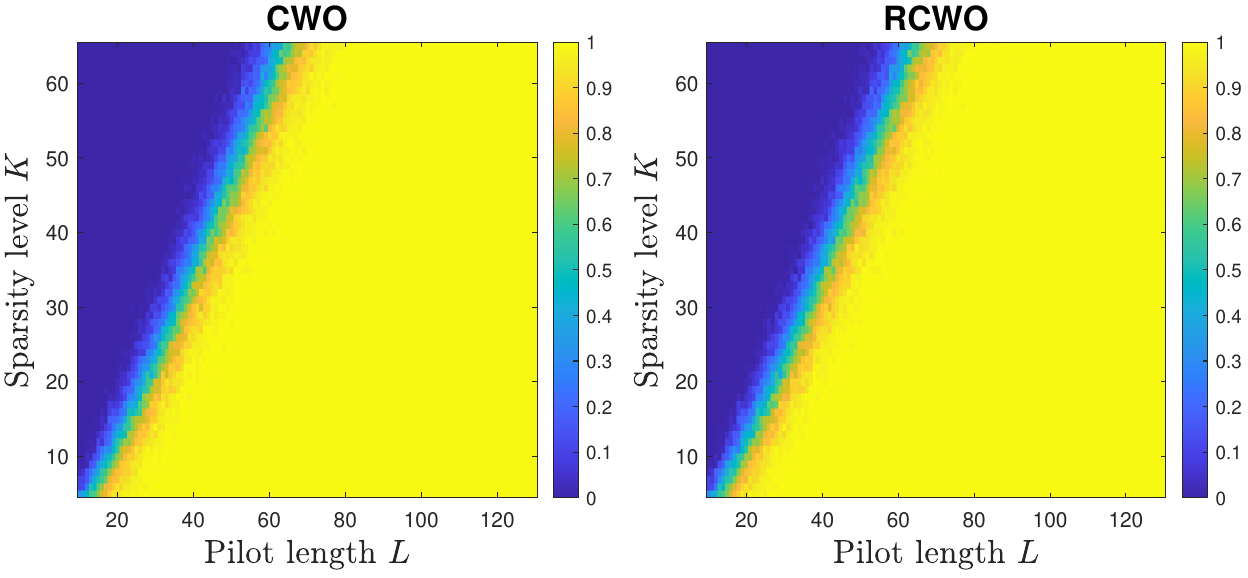}
    \caption{
        \revised{Phase transition of the average $P_\textrm{ER}$ for the CWO and RCWO algorithms, with $N=1000$ and $M=100$. 
        The noise follows a standard Gaussian distribution.
        The tuning parameter of Huber loss function in RCWO is set as $q = 0.999$.}
        }
    \vspace{-3mm}
    \label{fig:phase_gaussian}
\end{figure}
\begin{figure*}[!t]
    \centerline{\includegraphics[width = 0.95\linewidth]{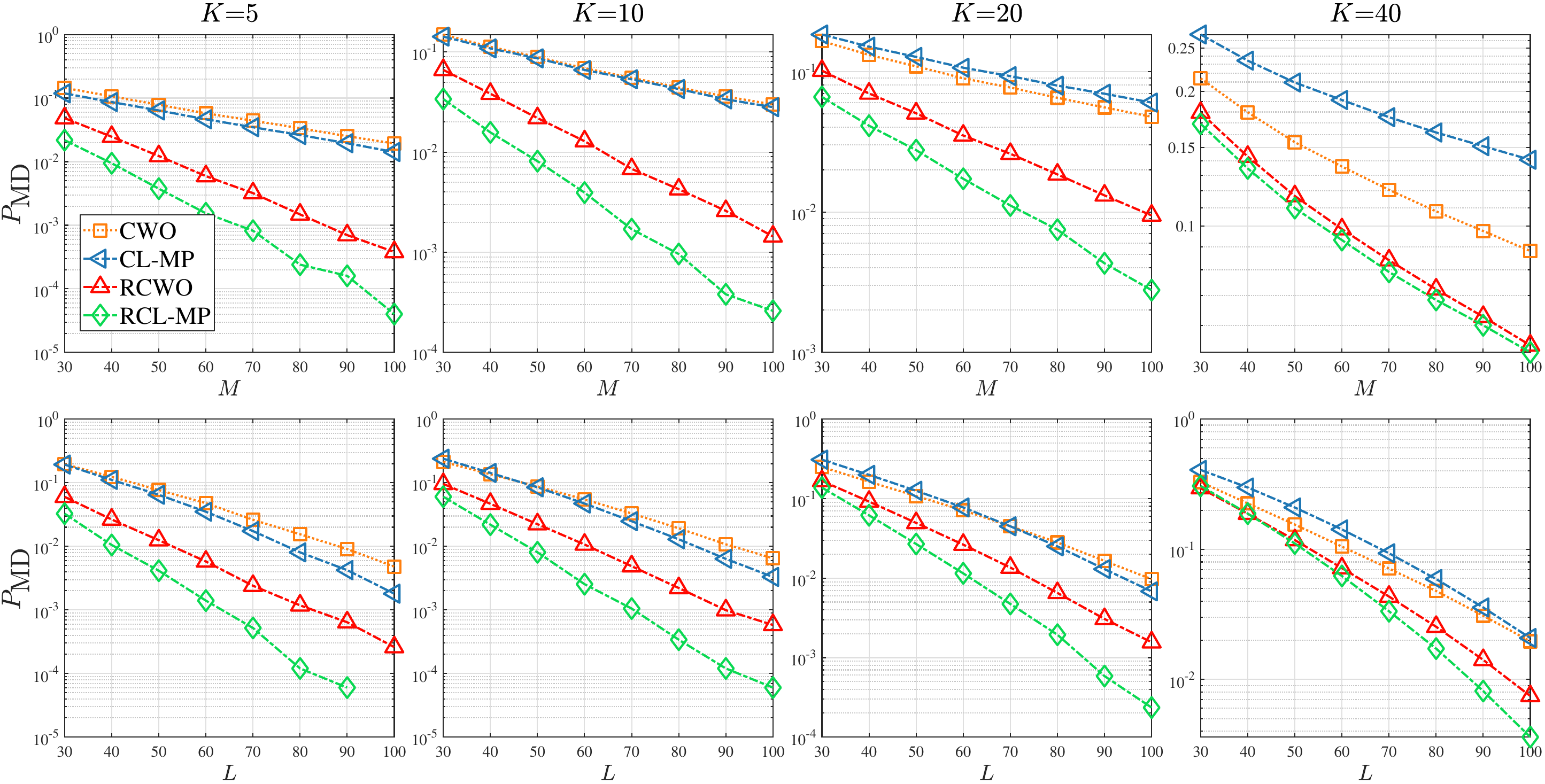}}
    \caption{\textit{Top:} Probability of Miss Detection ($\pmd$) versus the number of antennas $M$ with a fixed pilot length ($L=50$) for varying sparsity levels $K \in \{5,10,20,40\}$.
    \textit{Bottom:} $\pmd$ versus pilot length $L$ with a fixed number of antennas ($M=50$). The proposed RCWO and RCL-MP algorithms use Huber's loss with parameter $q=0.9$, and the noise follows an $\epsilon$-contaminated Gaussian distribution ($\epsilon=0.02$, $\lambda=10$).
    The number of devices is $N = 1000$.}
    \vspace{-3mm}
    \label{fig:synthetic_PmdvsMandL}
\end{figure*}
\begin{figure}[!t]
    \centerline{\includegraphics[width = 0.98\linewidth]{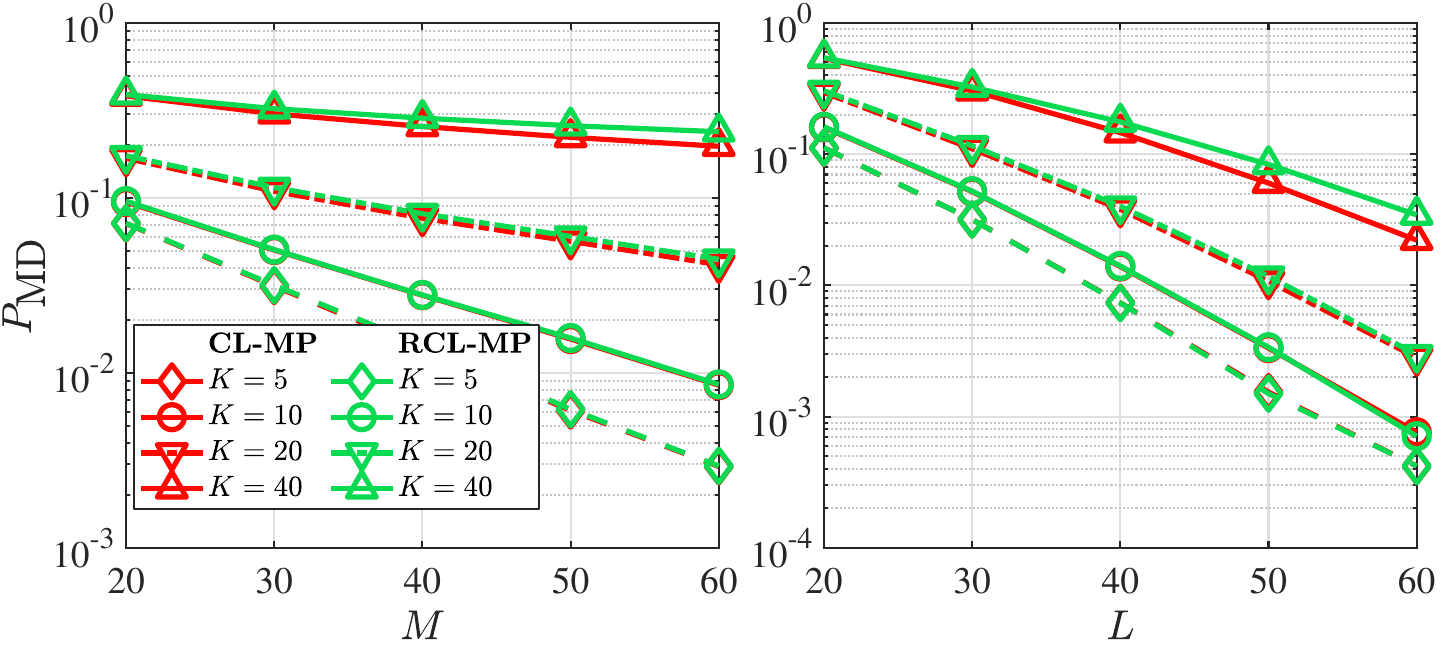}}
    \caption{Consistency between CL-MP and RCL-MP algorithms with tuning constant $q=0.999$ under $\epsilon$-contaminated noise, with $\epsilon=0$ indicating purely Gaussian noise. 
    The number of devices is $N=1000$, and the number of active devices vary within $\{5,10,20,40\}$.
    \textit{Left:} $\pmd$ versus varying number of antennas $M$ with a fixed pilot length ($L=30$).
    \textit{Right:} $\pmd$ versus varying pilot length $L$ with a fixed number of antennas ($M=30$).
    }
    \vspace{-3mm}
    \label{fig:PmdvsM_gaussian}
\end{figure}
\begin{figure*}[!t]
    \centering
    \includegraphics[width = 0.8\linewidth]{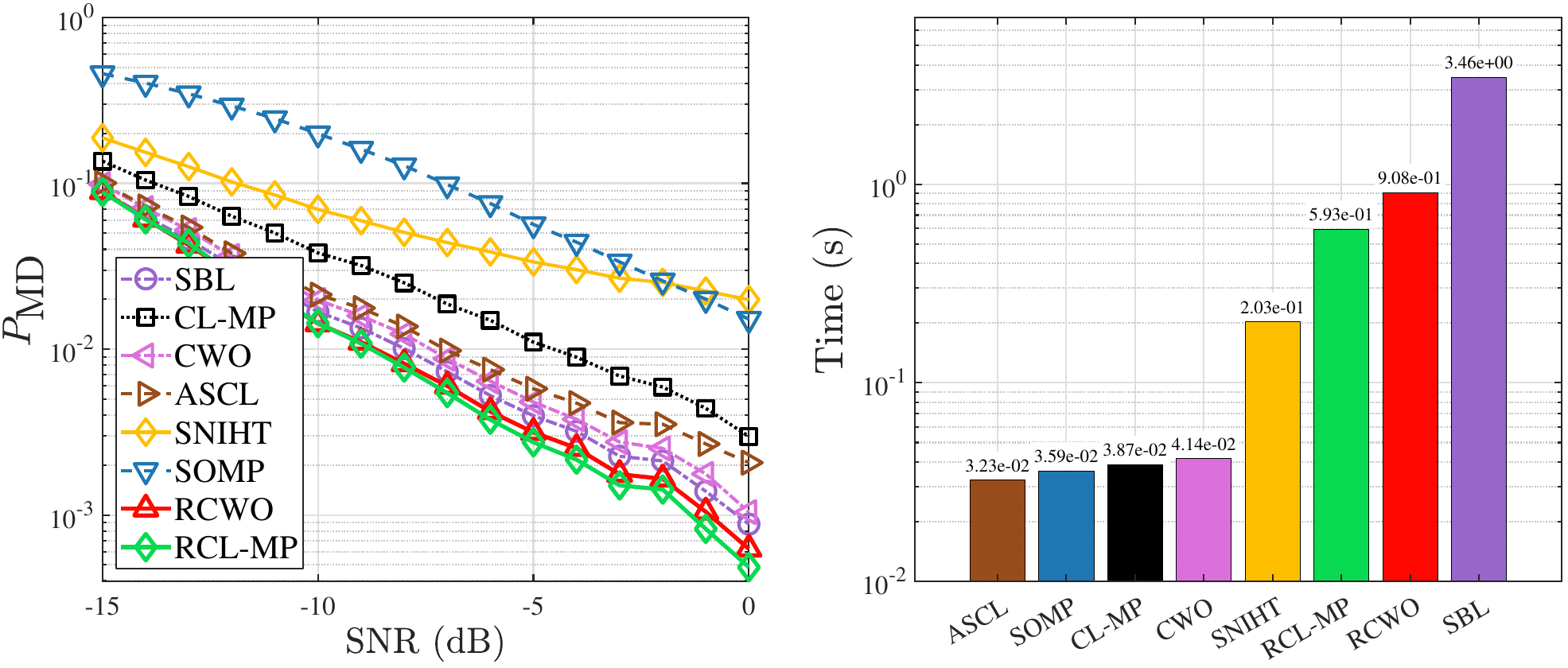}
    \caption{\textit{Left}: Probability of miss detection versus SNR. 
    \textit{Right}: Average running time of compared algorithms.
    The noise is impulsive $t$-noise with $\nu=2.5$, and the system configuration is as: $N=1000$, $K=25$, $M=40$ and $L=40$.}
    \label{fig:Uplink_PMDPFA}
\end{figure*}

In this simulation, the elements of the random pilot matrix~$\bbA$ are drawn from a Bernoulli distribution and normalized as $\|\bba_i\|_2^2= L$ for $i=1,\ldots,N$. 
The total number of devices is set to $N=1000$.
The LSFCs~$\{\beta_i\}_{i=1}^N$ are assumed to follow a uniform distribution in decibel (dB) scale, where~$10\log_{10}(\beta_i)$ ranges from $10\log_{10}(\beta_{\min})$ to $10\log_{10}(\beta_{\max})$~dB.
We select $\beta_{\min}$ and $\beta_{\max}$ such that this distribution ranges in $[-15, 0]$ dB.
We compare the proposed RCWO and RCL-MP algorithms with the CWO algorithm~\cite[Algorithm~1]{fengler2021nonTIT} and the CL-MP algorithm~\cite{marata2024activity}.
Both CL-MP and RCL-MP algorithms terminate after~\(K \) iterations, returning the estimated support by identifying the indices of the $K$-largest entries. 
The set $\mathcal M$ of active devices is randomly chosen from ${1,\ldots,N}$ for each Monte Carlo trial.

We test these methods under $\epsilon$-contaminated noises, where the noise $\bbe_m$ is drawn from $\mathcal C \ccalN_L (\bbzero, \sigma^2\bbI)$ with probability~$1 - \epsilon$ and from another Gaussian distribution $\mathcal C \ccalN (\bbzero, \lambda^2\sigma^2\bbI)$ with probability $\epsilon$, where $\lambda$ denotes the strength of outlier noises.
Note that for $\epsilon\to 0$, this noise becomes Gaussian. 
For example, if we set $\epsilon=0.02$, i.e., majority of measurements are corrupted by the nominal Gaussian noise with known variance~$\sigma^2$, while a small proportion ($2\%$) is corrupted by Gaussian noise with a tenfold larger amplitude (i.e., $\lambda=10$).
The loss function in RCWO and RCL-MP algorithms is set to the Huber's loss~\eqref{eq:loss-H}, and we set the loss parameter as $q=0.9$.

\subsubsection{Phase transition}
We examine the impact of the sparsity level~$K$ and pilot length $L$ on the probability of exact recovery ($\per$) for CWO, CL-MP, RCWO and RCL-MP in comparison.

\revised{
Fig.~\ref{fig:phase_transition} illustrates the phase transition for all compared algorithms under $\epsilon$-contaminated Gaussian noise. 
The experimental results demonstrate that our proposed algorithms (RCWO and RCL-MP) exhibit better robustness compared to their non-robust counterparts (CWO and CL-MP) under non-Gaussian noise. 
Specifically, the RCWO and RCL-MP algorithms achieve a higher probability of exact recovery ($\per$) than CWO and CL-MP when the noise is contaminated by outliers.

We also examine the case of standard Gaussian noise, shown in Fig.~\ref{fig:phase_gaussian}. When noise follows a Gaussian distribution, CWO and RCWO (with $q=0.999$) exhibit nearly identical phase transition patterns. 
This is expected since RCWO reduces to CWO when $q$ approaches 1 and the noise is Gaussian. 
This observation aligns with the theoretical analysis in~\cite{chen2021phase}, where the phase transition of CWO~\cite[Algorithm~1]{fengler2021nonTIT} under Gaussian noise was characterized. 
The similarity of phase transition patterns confirms that our robust framework preserves the theoretical guarantees established in~\cite{chen2021phase} for the Gaussian case.
}

\subsubsection{Effect of $M$, $L$ and $K$}
\label{subsubsec:syntheticexp_effectMLK}
We investigate the effects of number of antennas~$M$, pilot length~$L$, and number of active devices~$K$ on the performance of AD.
The number of Monte Carlo trials is~$3000$ and the total number of devices is~$N = 1000$.
Fig.~\ref{fig:synthetic_PmdvsMandL} shows the probability of miss detection $\pmd$ as a function of the number of antennas $M$ and the pilot length $L$, with varying numbers of active devices~$K \in \{5,10,20,40\}$.
In the \textit{top} panel, where the pilot length is fixed ($L=50$) and the number of antennas varies, our proposed RCWO and RCL-MP algorithms outperform the other covariance learning-based methods, CWO and CL-MP, which both assume Gaussianity. 
The probability of missdetection improves for all algorithms as $M$ increases due to enhanced structural sparsity in the received signal matrix $\bbY$, leading to more accurate estimations. 
This improvement is more pronounced for RCL-MP and RCWO algorithms, highlighting the robustness and benefits of a larger antenna array. 
As the number of active devices $K$ increases, all algorithms show reduced accuracy due to decreased sparsity. 
Notably, 
\revised{the performance of RCWO gradually approaches that of RCL-MP as $K$ increases},
indicating that the matching pursuit-based methods are more effective when the sparsity level is high, i.e., when $K$ is small.

In the \textit{bottom} panel, with the number of antennas fixed at~$M=50$ and varying pilot lengths $L$, performance improves as $L$ increases due to better pilot structure and longer detection periods. 
RCWO and RCL-MP consistently outperform the other Gaussian-based methods across all pilot lengths. 
The same trend is observed as $K$ increases, with RCWO eventually reaching same level of performance as RCL-MP when the number of active devices is the largest ($K=40$).

\subsubsection{Consistency with Gaussian-based method}

Fig.~\ref{fig:PmdvsM_gaussian} highlights the strong similarity between RCL-MP and its Gaussian counterpart, CL-MP, when using Huber's loss with tuning parameter set nearly to its maximal value ($q=0.999$) under Gaussian noise, across varying sparsity levels ($K \in \{5, 10, 20, 40\}$). 
The curves for CL-MP and RCL-MP overlap in most scenarios, suggesting that the RCL-MP algorithm provides equivalent performance as CL-MP when $q\approx 1$ as expected. 
This adaptability indicates that our proposed method can effectively handle both Gaussian and non-Gaussian noise, demonstrating greater robustness and flexibility compared to methods that assume Gaussian noise exclusively.

\subsection{Cellular MIMO uplink simulation}
\begin{figure*}[!t]
    \centering
    \subfloat{
        \includegraphics[width = 0.9\linewidth]{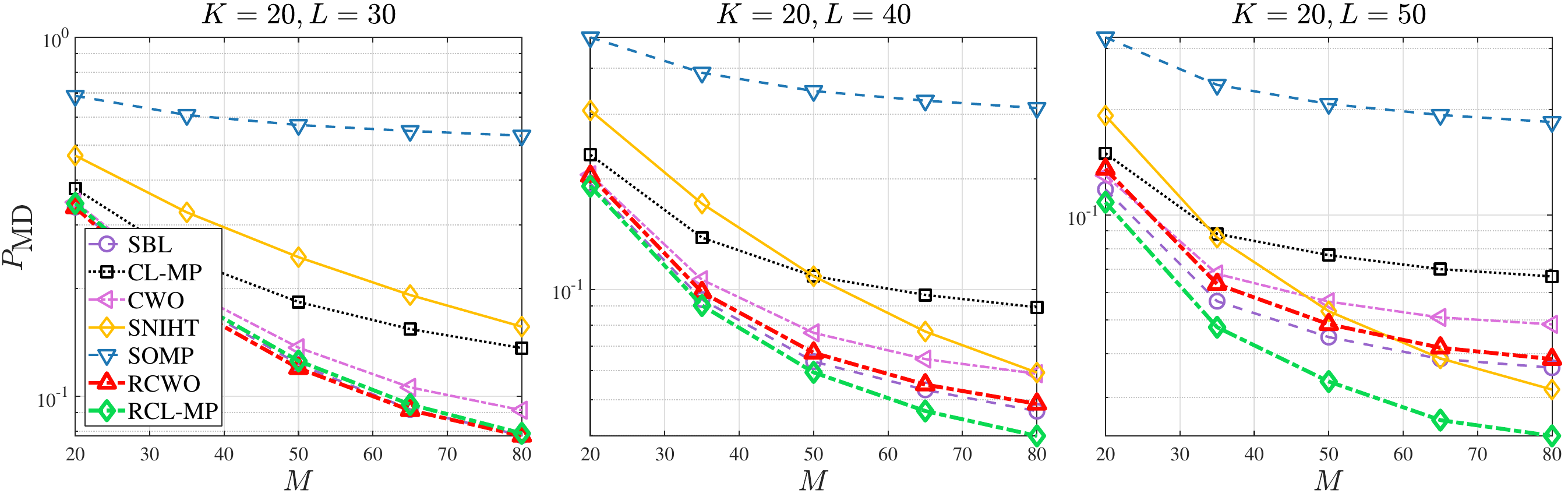}}\\
    \vspace{-2mm}
    \subfloat{
        \includegraphics[width = 0.9\linewidth]{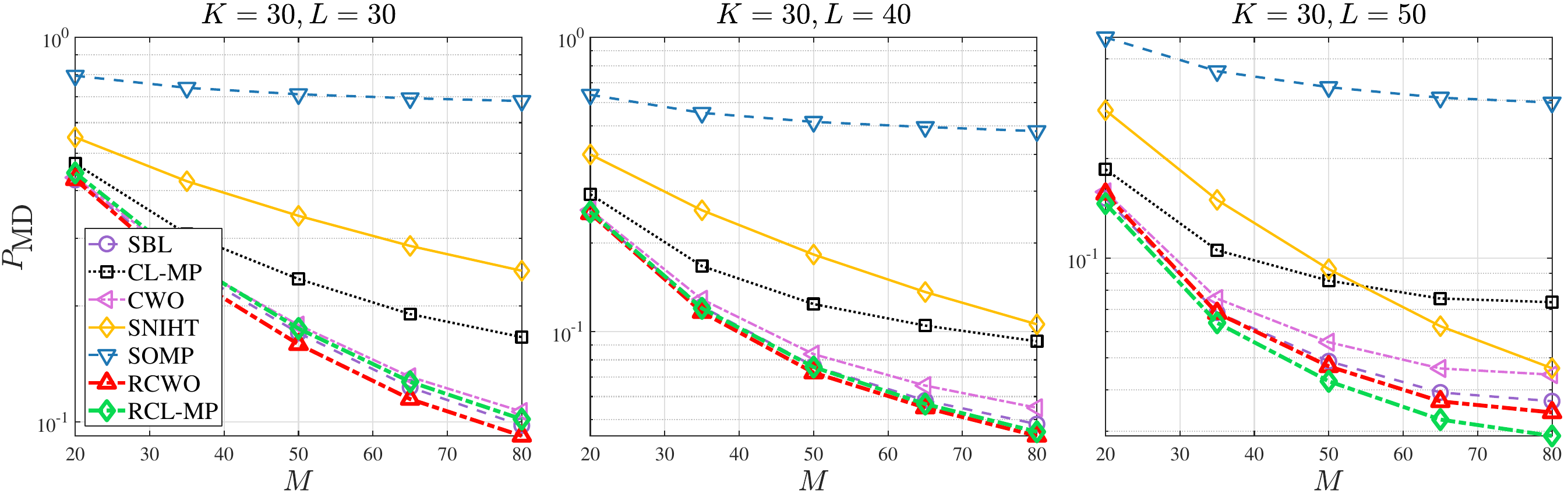}}\\
    \vspace{-2mm}
    \subfloat{
        \includegraphics[width = 0.9\linewidth]{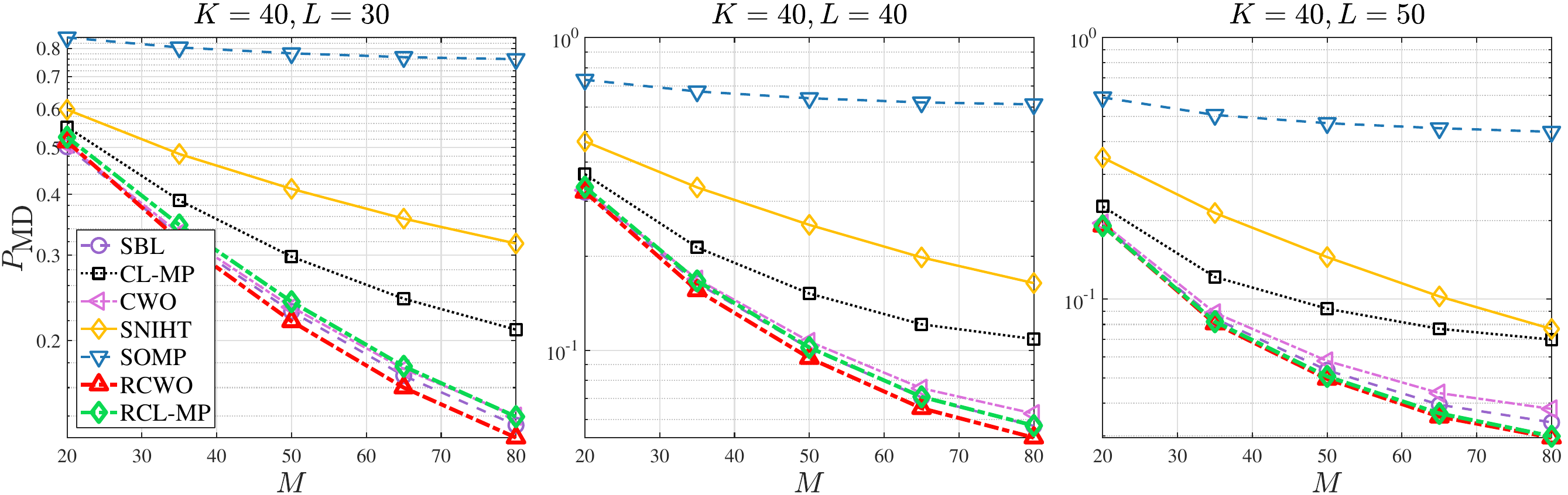}}
    \caption{Impact of the number of antennas $M$ on the probability of miss detection ($\pmd$) for various sparsity levels $K$ and pilot lengths $L$. Rows correspond to different sparsity levels $K \in \{20, 30, 40\}$, and columns represent pilot lengths $L \in \{30, 40, 50\}$. The noise is impulsive $t$-noise with $\nu=2.5$. We set $\SNR=-15\textup{dB}$ and $N=1000$. 
    \revised{
    Note that the SBL curve (purple) overlaps with the RCL-MP curve in the top-left panel ($K=20, L=30$) and bottom-middle panel ($K=40, L=40$) due to their similar performance in these specific configurations.
    }
    }
    \vspace{-3mm}
    \label{fig:MIMOuplink_overKLM}
\end{figure*}
We evaluate the performance of the proposed RCWO and RCL-MP algorithms in a cellular massive multiple-input and multiple-output (MIMO) uplink mMTC network. 
The BS serves $N=1000$ devices, which are randomly distributed within a range of $[0.05 \text{ km}, 1 \text{ km}]$.
For the LSFCs, we use the path loss model from~\cite{liu2018sparse}, expressed as 
\begin{align} 
    \beta_{i} = -128.1 - 36.7 \log_{10}(d_{i}), 
\end{align} 
where $d_{i}$ represents the distance in kilometers between the BS and the $i$-th device.
In terms of power control, we follow the scheme proposed in~\cite{senel2018grant}
\begin{align} 
    \varrho_{i} = \frac{\varrho_{\max} (\beta_{i})_{\min}}{\beta_{i}}, 
\end{align} 
where $\varrho_{\max}$ is the maximum transmission power of the devices, and $(\beta_{i})_{\min}$ denotes the minimum LSFC. 
This scheme ensures that the device with the lowest LSFC transmits at maximum power, while the transmission powers of the other devices scale inversely with their LSFCs. 
We set the maximum transmission power to $\varrho_{\max} = 0.2$W.
The bandwidth and coherence time of the wireless channel are set to $1$ MHz and $1$ ms, respectively, allowing for the transmission of $1000$ symbols. 
We allocate $10\%$ of the available symbols for pilot sequences, resulting in a maximum pilot length of $L_{\max} = 100$.

We use the Bernoulli pilot as described in Section~\ref{subsec:synthetic_simulation}.
Due to the finite set of unique pilot sequences, there is a non-zero probability of two devices having identical sequences, known as the collision probability~\cite{senel2018grant}. 
According to the calculation in~\cite[Eq.~(2)]{senel2018grant}, for $N=1000$ devices and a pilot length of $L=20$, the collision probability is $4.54 \times 10^{-7}$, which is negligible in practice.
As the pilot length increases while keeping the number of devices fixed, the collision probability further decreases.

The performance of the proposed RCWO and RCL-MP algorithms is compared to the following algorithms
\begin{enumerate}
    \item Coordinatewise Optimization Algorithm (CWO)~\cite[Algorithm 1]{fengler2021nonTIT}
    \item Covariance-based Matching Pursuit (CL-MP)~\cite{marata2024activity}
    \item Sparse Bayesian Learning (SBL)~\cite{wipf2004sparse}
    \item Simultaneous Normalized Iterative Hard Thresholding (SNIHT)~\cite[Algorithm 1]{blanchard2014greedy}
    \item Simultaneous orthogonal matching pursuit (SOMP)~\cite[Algorithm 3.1]{tropp2006algorithms}
    \item Active Set Covariance Learning (ASCL)~\cite[Algorithm~1]{wang2021efficientJoint} 
\end{enumerate}
Note that the SBL algorithm is implemented using the Expectation-Maximization (EM) approach as described in~\cite{wipf2004sparse}, rather than the coordinate-wise optimization variant.

\revised{
    Each row vector of the noise matrix $\bbE \in \mathbb{C}^{L \times M}$ is generated independently from a circular $M$-variate $t$-distribution with $\nu=2.5$  degrees of freedom and covariance matrix $\sigma^2\bbI_{M}$. 
    The $t$-distribution naturally models impulsive noise through its heavier tails compared to Gaussian distributions, with the parameter $\nu$ controlling  heaviness of the tails (smaller values produce more frequent outliers). 
    For $\nu > 2$, this noise can be generated by scaling an $M$-variate circular complex Gaussian random vector $ \boldsymbol{z} \sim \mathcal C \mathcal N_M(\boldsymbol{0}, \bbI)$  by $\sqrt{\sigma^2(\nu-2)/\chi^2_{\nu}}$, where $\chi^2_{\nu}$ represents independent chi-squared random variable with $\nu$ degrees of freedom. 
    Our choice of $\nu = 2.5$ represents moderately heavy-tailed noise that effectively challenges Gaussian-based methods while maintaining finite variance.
    }

\subsubsection{Effect of SNR}
In this simulation, the number of active devices is set to \( K = 25 \), with the number of antennas and pilots both set to \( M = 40 \) and \( L = 40 \), respectively. 
The performance is evaluated using \( \pmd \), along with the average running time of the algorithms.

Fig.~\ref{fig:Uplink_PMDPFA} (left panel) shows the \( \pmd \) performance of all algorithms over SNR levels ranging from \([-15, 0]\)~dB.
The proposed RCL-MP algorithm achieves the best performance in all cases, closely followed by the robust RCWO.  
As the SNR increases, \( \pmd \) decrease. 
The SBL algorithm ranks third, and thus has the best performance among the competing methods. 
However, SBL's quadratic complexity (\( \ccalO(N^2L) \)) with respect to the total number of devices \( N \) makes it impractical for large-scale mMTC scenarios, where \( N \) is typically large. 
The proposed RCWO and RCL-MP algorithms consistently outperform their Gaussian-based counterparts at all SNR levels.

Fig.~\ref{fig:Uplink_PMDPFA} (right panel) shows the average running time of the algorithms. 
The ASCL algorithm shows slightly faster running time compared to CWO while maintaining slightly worse detection performance.
The running time of RCL-MP is comparable to SNIHT and Gaussian-based CWO. 
RCWO is slower than RCL-MP but faster than SBL, which has the longest running time due to its high computational complexity.
The additional running time of RCL-MP compared to CL-MP is caused by the FP computations in the FP-CW algorithm. 
CL-MP, SOMP and SNIHT show lower running times but underperform in~\( \pmd \) compared to Gaussian-based CWO and the proposed robust methods. 
While methods such as ASCL, CL-MP, and SOMP show the lowest running times (yet comparable to that of the proposed CWO), they consistently underperform in~\(\pmd\) compared to the proposed robust methods.
While RCL-MP requires more computation, it achieves better performance under non-Gaussian noise and maintains a running time comparable to CWO.

\subsubsection{Effect of $M$, $L$ and $K$}
We investigate the effects of the number of antennas~$M$, pilot length~$L$, and the number of active devices~$K$ on the AD performance. 
Unlike the setup in Section~\ref{subsubsec:syntheticexp_effectMLK}, this evaluation employs a different configuration of LSFCs, uses impulsive $t$-noise with $\nu=2.5$, and assumes an average received SNR of $-15$ dB.
Fig.~\ref{fig:MIMOuplink_overKLM}~illustrates $\pmd$ as a function of the number of antennas $M$ across various combinations of sparsity levels $K$ and pilot lengths $L$. 
In this figure, rows correspond to varying sparsity levels  $K \in \{20, 30, 40\}$, while columns represent different pilot lengths $L \in \{30, 40, 50\}$.
For all methods, increasing the number of antennas $M$ reduces $\pmd$ consistently, as more antennas enhance the structural sparsity of the received signal $\bbY$, improving estimation accuracy. 
Fig.~\ref{fig:MIMOuplink_overKLM} shows that this trend holds for all pilot lengths $L$ and sparsity levels $K$.

As the pilot length \( L \) increases, the overall performance improves for all algorithms, as shown by the progression from the left to right columns in Fig.~\ref{fig:MIMOuplink_overKLM}. 
This improvement is due to longer pilot lengths and detection periods, which enhance covariance estimation.
This behavior is consistent with the Gaussian noise case, where CL-based methods exhibit quadratic performance scaling with pilot length~\cite{chen2021phase}. However, under impulsive noise, the performance of Gaussian based algorithms changes.

Within each column, the relative performance of RCWO and RCL-MP varies. 
As \( K \) increases, the relative advantage of RCL-MP diminishes. 
For high sparsity levels (\( K = 40 \)) and short pilot lengths (\( L = 30 \)), RCL-MP does not achieve the best performance, while the CWO-based methods (RCWO and CWO) perform better, with RCWO outperforming CWO. However, when $L$ increases to $L=40$ and $L=50$, RCL-MP performs better than CWO. 
This is likely due to the additional errors introduced into the covariance estimation under strong impulsive noise (\( \text{SNR} = -15 \)~dB) when longer pilot lengths are used. 
In contrast, matching pursuit-based methods, such as RCL-MP, operate over a fixed \( K \) iterations, making them less susceptible to noise accumulation. 
As \( L \) increases, RCL-MP outperforms all other methods, while RCWO’s relative performance declines.
These results demonstrate the trade-offs between pilot length, sparsity levels, and robustness under impulsive noise. While RCL-MP is generally robust and performs well across scenarios, its relative advantage depends on the interplay of \( M \), \( L \), and \( K \).

\section{Conclusion}
\label{sec:Conclusion}
This paper presented robust algorithms for device AD in massive mMTC under non-Gaussian noise conditions. 
By replacing the Gaussian loss function with robust alternatives, we addressed the limitations of existing methods in handling heavy-tailed noise and outliers. 
The proposed methods, supported by a FP algorithm with proven convergence, demonstrated superior performance in extensive numerical experiments, particularly in challenging impulsive noise environments.
The numerical experiments confirmed the effectiveness and robustness of the proposed approaches in improving detection accuracy under diverse noise conditions, while maintaining practical computational efficiency. 
These findings make the proposed methods promising solutions for real-world mMTC systems where accurate activity detection is critical.
Future work may explore extending these methods to more complex network settings, such as cell-free architectures, and addressing real-time implementation challenges in large-scale systems.

\appendix
\subsection{Proof of Theorem~\ref{thm:gconvexity_Li}} 
\label{proof:thm_gconvexity_Li}
Recall the conditional objective function
\begin{align*}
    \label{eq:condi_proof} 
    \ccalL_i(\gamma) = 
    \frac1{bM}\sum_{m=1}^{M} \rho(s_{m,i}(\gamma))
    + \log|\bSb + \gamma\bba_i\bba_i^{\herm}|.
\end{align*}
Our objective is to show that $\ccalL_i(\gamma)$ is geodesically convex in~$\gamma\geq 0$, meaning that $\ccalL_i(e^x)$ is convex in $x\in\reals$.
It suffices to show that $\rho(s_{m,i}(\gamma))$ and the 2nd term of $\ccalL_i(\gamma)$ and are geodesically convex.
Since the sum of geodesically convex functions is still geodesically convex, the entire objective function $\ccalL_i(\gamma)$ is geodesically convex in~$\gamma\in\reals_{\geq0}$.

Let
\begin{align}
    h(\gamma) = \log|\bSb + \gamma\bba_i\bba_i^{\herm}|.
\end{align}
Using the matrix determinant lemma\footnote{The matrix determinant lemma for a rank-one update states that for an invertible matrix $\bbA$ and vectors $\bbu,\bbv$, the determinant of $\bbA + \bbu\bbv^{\herm}$ is given by $|\bbA + \bbu\bbv^{\herm}| = (1 + \bbv^\herm \bbA^\inv \bbu)|\bbA|$.}  and the shorthand notation $d_i = \bba_i^{\herm} \bSb^\inv \bba_i$, we have
\begin{align}
    h(e^x) & = \log|\bSb + e^x \bba_i\bba_i^{\herm}| \notag \\
    & = \log \left( (1 + e^x \bba_i^{\herm} \bSb^\inv \bba_i) |\bSb| \right) \notag \\
    & = \log(1 + d_i e^x) + \log |\bSb|.
\end{align}
The second derivative of $h(e^x)$ is
\begin{align}
    h''(e^x)
    = \frac{d_i e^x}{(1+ d_i e^x)^2} \geq 0,
\end{align}
which shows that $h(e^x)$ is convex in $x \in \mathbb{R}$, or  equivalently, $h(\gamma)$ is geodesically convex in $\gamma\in\reals_{\geq0}$.

From the known expression for $s_{m,i}(\gamma)$ in~\eqref{eq:fp_alg_1eta}, we write 
\begin{align*}
s_{m,i}(\gamma) 
    & = c_m - \frac{\gamma b_m^2} {1+\gamma d_i},
\end{align*}
where  $b_m = |\bby_m^{\herm}\bbb_i|$ and $c_m = \bby_m^{\herm}  \bSb^{-1}  \bby_m$. 
We need to prove that $\rho(s_{m,i}(\gamma))$ is geodesically convex in $\gamma$, i.e., the function 
\[
g(\gamma)= \rho\left(c_m -   b_m^2\gamma (1+\gamma d_i)^{-1} \right)
\]
is geodesically convex. 
Define a function 
\begin{equation*} 
    \tilde g(\gamma,\eta)= \rho\left(c_m -   b_m^2\gamma \eta \right).
\end{equation*}
Let  $\eta=(1+\gamma d_i)^{-1}$ and note that $\eta \in \reals_{\geq0}$.
Write 
\begin{equation}\label{eq:tildeg}
    \tilde g(e^x,e^y)=  \tilde \rho\left(e^{x+y} \right),
\end{equation}
where $\tilde \rho(\gamma)=\rho\left(c_m -   b_m^2  \gamma \right)$ and $x,y \in \mathbb{R}$. 
Since $\rho(\gamma)$ is geodesically convex, also $\tilde \rho(\gamma)$ is geodesically convex in $\gamma$. 
This together with \eqref{eq:tildeg} implies that $\tilde g(\gamma,\eta)$ is jointly geodesically convex in both of its arguments. 
This implies geodesic convexity of $g(\gamma)$. 

Thus, $\frac1M\sum_{m=1}^M \rho (s_{m,i}(\gamma))$ is geodesically convex in $\gamma \in \reals_{\geq0}$.
This completes the proof.

\revised{
\subsection{Proof of Theorem~\ref{thm:globalminimum}}
\label{proof:thm_globalminimum}
\text{(a)}
This follows directly from Theorem~\ref{thm:gconvexity_Li}, as any geodesically convex function has the property that any local minimum is a global minimum.

\text{(b)}
Since $\rho(\cdot)$ is bounded below, the first term of $\mathcal{L}_i(\gamma)$, $\frac{1}{bM}\sum_{m=1}^M \rho(s_{m,i}(\gamma))$, is bounded below. 
The second term, $\log |\bSb + \gamma\bba_i\bba_i^\herm|$, is also bounded below with a finite minimum value $\log |\bSb|$ at $\gamma = 0$. Therefore, $\mathcal{L}_i(\gamma)$ is bounded below.

Since $\rho(\cdot)$ is non-decreasing, and the second term of $\mathcal{L}_i(\gamma)$ satisfies
\begin{align*}
    \log |  \bSb +   \gamma \bba_i\bba_i^{\herm} | \to \infty \ \ \text{as} \ \ \gamma \to +\infty,
\end{align*}
we conclude that $\mathcal{L}_i(\gamma) \rightarrow \infty$ as $\gamma \rightarrow +\infty$, i.e., $\mathcal{L}_i(\gamma)$ is coercive. 
Thus, $\{\gamma^{(\ell)}\}_{\ell}$ is bounded.
It follows from geodesic convexity of the objective, that any accumulation point of $\{\gamma^{(\ell)}\}_{\ell}$ is a global minimizer of $\mathcal L_i(\gamma)$.
The combination of being bounded below and coercive ensures the existence of at least one global minimizer.

\text{(c)}
When $\rho(\cdot)$ is strictly geodesically convex, $\mathcal L_i(\gamma)$ is also strictly geodesically convex and thus has a unique global minimum.
Hence, the sequence $\{ \gamma^{(\ell)}\}_{\ell}$ converges to the unique accumulation point, which is the unique global minimizer of~$\mathcal{L}_i(\gamma)$.
}

\subsection{Proof of Theorem~\ref{thm:monotonicityofsequence}}
\label{apdx:monotonicity}
Using the concavity of \( \rho(\cdot) \), we have
\begin{align}
    \rho(\eta) \leq \rho(\eta^0) + (\eta - \eta^0) \rho'(\eta^0), \quad \forall \eta, \eta^0.
\end{align}
Applying \( \eta = s_{m,i}(\gamma) \) and \( \eta^0 = s_{m,i}(\gamma^{(\ell)}) \) yields
\begin{align}
    & \rho(s_{m,i}(\gamma)) \notag \\
    & \leq \rho(s_{m,i}(\gamma^{(\ell)})) + \rho'(s_{m,i}(\gamma^{(\ell)}))(s_{m,i}(\gamma) - s_{m,i}(\gamma^{(\ell)})) \notag \\
    & = \rho(s_{m,i}(\gamma^{(\ell)})) + u_{m,i}(\gamma^{(\ell)})(s_{m,i}(\gamma) - s_{m,i}(\gamma^{(\ell)})),
\end{align}
where \( u_{m,i}(\gamma^{(\ell)}) = \rho'(s_{m,i}(\gamma^{(\ell)})) \).  

Thus, we can construct the following surrogate function for the conditional negative log-likelihood function~\eqref{eq:cond_nllf} as
\begin{align}
    & Q_i(\gamma \mid \gamma^{(\ell)}) = \frac{1}{bM} \sum_{m=1}^M \Big( \rho(s_{m,i}(\gamma^{(\ell)})) \notag \\
    & + u_{m,i}(\gamma^{(\ell)})(s_{m,i}(\gamma) - s_{m,i}(\gamma^{(\ell)})) \Big) + \log|\bSb + \gamma\bba_i\bba_i^{\herm}|.
\end{align}
This surrogate function satisfies
\begin{subequations}
    \begin{align}
        Q_i(\gamma \mid \gamma^{(\ell)}) & \geq \ccalL_i(\gamma), \label{eq_mono_proof:MM1}\\
        Q_i(\gamma^{(\ell)} \mid \gamma^{(\ell)}) & = \ccalL_i(\gamma^{(\ell)}). \label{eq_mono_proof:MM2}
    \end{align}
\end{subequations}

Taking the first derivative of \( Q_i(\gamma \mid \gamma^{(\ell)}) \) with respect to \( \gamma \), we have
\begin{align}
    Q'_i(\gamma \mid \gamma^{(\ell)}) = \frac{1}{bM} \sum_{m=1}^M u_{m,i}(\gamma^{(\ell)}) \frac{\partial s_{m,i}(\gamma)}{\partial \gamma} + \frac{d_i}{1 + \gamma d_i}.
\end{align}
Setting this derivative to zero, we obtain the minimizer of \( Q_i(\gamma \mid \gamma^{(\ell)}) \),
\begin{align}
    \hhatgamma &= \frac{ \frac{1}{bM} \sum_{m=1}^M u_{m,i}(\gamma^{(\ell)}) |\bby_m^{\herm}\bbb_i |^2 - d_i}{d_i^2}.
\end{align}
Since this update is equivalent to the FP iteration step~\eqref{eq_fp_alg:4_nonnegative_gamma}, we have \( \gamma^{(\ell+1)} = \hhatgamma \).

Since \( \gamma^{(\ell+1)} \) minimizes \( Q_i(\gamma \mid \gamma^{(\ell)}) \), we have
\begin{align}
    Q_i(\gamma^{(\ell+1)} \mid \gamma^{(\ell)}) \leq Q_i(\gamma^{(\ell)} \mid \gamma^{(\ell)}).
\end{align}
Combining this with \eqref{eq_mono_proof:MM1} and \eqref{eq_mono_proof:MM2}, we have
\begin{align}
    \ccalL_i(\gamma^{(\ell+1)}) \leq Q_i(\gamma^{(\ell+1)} \mid \gamma^{(\ell)}) \leq Q_i(\gamma^{(\ell)} \mid \gamma^{(\ell)}) = \ccalL_i(\gamma^{(\ell)}).
\end{align}
Thus, we have shown that \eqref{eq:proof:inequality} holds.
\revised{If $\rho(\eta)$ is increasing, the strict concavity, $\rho(\eta) < \rho(\eta^0) + (\eta - \eta^0)\rho'(\eta^0)$, holds. It implies than that \eqref{eq:proof:inequality} holds as stict inequality, i.e., $ \ccalL_i(\gamma^{(\ell+1)})< \ccalL_i(\gamma^{(\ell)})$.}
This completes the proof. 

\bibliographystyle{IEEEtran}
\bibliography{IEEEabrv, reference}
\end{document}